\documentclass[journal,10pt,onecolumn,draftclsnofoot,]{IEEEtran} 
\usepackage{graphicx}
\usepackage[english]{babel}
\usepackage[utf8x]{inputenc}
\usepackage[T1]{fontenc}

\usepackage{amsmath}
\usepackage{amsthm}
\usepackage{amssymb}
\usepackage{graphicx}
\usepackage{mwe}
\usepackage[colorinlistoftodos]{todonotes}
\usepackage[colorlinks=true, allcolors=blue]{hyperref}
\usepackage{xspace,syntonly}
\usepackage{url}
\usepackage{algorithm,algorithmic}
\usepackage{mathtools}
\usepackage{cite}
\usepackage{pstool}
\usepackage{multirow}

\usepackage{times,inconsolata}
\usepackage{graphicx}
\usepackage{subcaption}
\usepackage{pifont}
\usepackage{fancyhdr}
\usepackage{mathtools}   
\usepackage{enumitem} 
\newtheorem{corollary}{Corollary}[]
\newtheorem{lemma}[]{Lemma}
\newtheorem{assumption}[]{Assumption}
\newtheorem{theorem}[]{Theorem}

\usepackage{cite}
\usepackage{environ}
\usepackage{algorithm}
\usepackage{algorithmic}
\makeatletter
\NewEnviron{Lalign}{\tagsleft@true\begin{align}\BODY\end{align}}
\makeatother
\newcommand{\utwi}[1]{\mbox{\boldmath $#1$}}

\newcommand{\bflambda}{{\mbox{\boldmath $\lambda$}}}

\newcommand{\cC}{{\mbox{$\mathcal{C}$}}}

\newcommand{\cN}{{\cal N}}

\newcommand{\cI}{{\cal I}}

\newcommand{\cB}{{\cal B}}

\newcommand{\cY}{{\cal Y}}

\newcommand{\bc}{{\bf c}}

\newcommand{\ba}{{\bf a}}
\newcommand{\bb}{{\bf b}}
\newcommand{\bm}{{\bf m}}

\newcommand{\bx}{{\bf x}}

\newcommand{\bv}{{\bf v}}
\newcommand{\bw}{{\bf w}}

\newcommand{\bz}{{\bf z}}
\newcommand{\by}{{\bf y}}
\newcommand{\bA}{{\bf A}}
\newcommand{\bB}{{\bf B}}

\newcommand{\bG}{{\bf G}}

\newcommand{\bM}{{\bf M}}

\newcommand{\bI}{{\bf I}}

\newcommand{\bell}{{\utwi{\ell}}}

\newcommand{\bbeta}{{\utwi{\beta}}}

\newcommand{\bgamma}{{\utwi{\gamma}}}

\newcommand{\blambda}{{\utwi{\lambda}}}

\newcommand{\sfT}{\textsf{T}}
\newcommand{\reals}{\mathbb{R}}

\newcommand{\RomanNumeralCaps}[1]
{\MakeUppercase{\romannumeral #1}}

\newcommand{\xmark}{\ding{55}}%

\usepackage{multirow}
\usepackage{pgfplots}

\allowdisplaybreaks

\begin{document}
\title{Online Proximal-ADMM For Time-varying \\ Constrained Convex Optimization}
\author{Yijian Zhang, Emiliano Dall'Anese \emph{Member, IEEE}, and Mingyi Hong \emph{Member, IEEE}
\thanks{Y. Zhang is with department of IMSE, Iowa State University, Ames, USA; email: yijian@iastate.edu. E. Dall'Anese is with department of Electrical, Computer, and Energy
Engineering, University of Colorado Boulder, Boulder, USA; email: emiliano.dallanese@colorado.edu. The work of E. Dall'Anese was supported in part by the National Science Foundation award 1941896. M. Hong is with department of ECE, University of Minnesota, Minneapolis, USA; email: mhong@umn.edu. Mingyi Hong was supported, in part, by the National Science Foundation (grant CIF-1910385, grant CNS-2003033).}
}

\maketitle

\begin{abstract}
This paper considers a convex optimization problem with cost and constraints that evolve over time. The function to be minimized  is  strongly convex and possibly non-differentiable, and  variables are coupled through linear constraints.  In this setting, the paper proposes an online algorithm based on the alternating direction method of multipliers (ADMM), to track the optimal solution trajectory of the time-varying problem; in particular, the proposed algorithm consists of a primal proximal gradient descent step and an appropriately perturbed dual ascent step. 
The paper derives tracking results, asymptotic bounds, and linear convergence results.The proposed algorithm is then specialized to a multi-area power grid optimization problem, and our numerical results verify the desired properties.

\end{abstract}


\IEEEpeerreviewmaketitle

\section{Introduction}
This paper considers  time-varying optimization problems for network systems, where objective and constraints evolve over time~\cite{popkov2005gradient,simonetto2014double,Shahrampour2018,hall2015online,bernstein2018online,simonetto2016class}. The applicability of time-varying optimization problems is evident in a number of domains including power grids~\cite{dall2018optimal,tang2017real,hauswirth2016projected}, communication systems~\cite{Low1999Flow,Chen12}, and online methods in signal processing \cite{asif2014sparse}, just to name a few; see also the representative works~\cite{rahili2017distributed,fazlyab2016self,neely2017online} and~\cite{dall2019optimization} for additional time-varying models and application examples.

In particular, assume that the temporal domain is discretized as $\{k \tau, k \in \mathbb{N}\}$, with $\tau>0$ being a given interval time~\cite{simonetto2016class,dall2019optimization}. This paper focuses on time-varying optimization problems in the following form~\cite{popkov2005gradient,simonetto2014double}:
\begin{subequations}
\label{P1}
\begin{align}
\textrm{(P1)~~} \min_{\bx\in\mathbb{R}^m,\by\in\mathbb{R}^n}~&f^{(k)}(\bx) + g^{(k)}(\by)\label{P1C}\\
\text{s.t.}~&\bA^{(k)}\bx+\bB^{(k)}\by=\bb^{(k)}\label{constraint_intro}\\
&\bx\in\mathcal{X}^{(k)}, \by\in\mathcal{Y}^{(k)},
\end{align}
\end{subequations}
where $k$ is the index for the time steps $\{k \tau, k \in \mathbb{N}\}$;  $f^{(k)}(\cdot)$ and $g^{(k)}(\cdot)$ are strongly convex functions (for all times $k \tau$); $\bA^{(k)}\in\mathbb{R}^{\ell\times m}, \bB^{(k)}\in\mathbb{R}^{\ell\times n}$ are time-varying matrices; and, $\mathcal{X}^{(k)},\mathcal{Y}^{(k)}$ are convex constraint sets. For optimization and learning problems with streams of data, the interval $\tau$ coincides with the inter-arrival date of data points~\cite{dall2019optimization}; when problem \eqref{P1} is associated with a network, the functions $f^{(k)}(\cdot)$ and $g^{(k)}(\cdot)$ can capture performance objectives that evolve over time, whereas~\eqref{constraint_intro} can capture time-varying physical or logical interactions in the network. Denoting as $\bx^{*,(k)}$ an optimal solution of~\eqref{P1} at time $k \tau$, the optimization model~\eqref{P1} leads to an optimal \emph{trajectory}. The problem addressed in this paper pertains to the development of algorithms that enable \emph{tracking of the optimal trajectory} $\{\bx^{*, (k)}\}_{k \in \mathbb{N}}$. It is worth noticing that, for  problems with inequality constraints, one can always add slack variables to re-write them as equality constraints -- thus fitting  the formulation \eqref{P1}. 


Previous efforts\cite{kar2011gossip,olfati2007consensus} have addressed dynamic (consensus-type) problems as a series of static problems, and have assumed a time-scale separation between algorithms and variability of the problem,  so that convergence is reached for each time. However, this might not be  when the system parameters and the problem inputs change fast, at a time scale that is comparable with the execution of one (or a few) algorithmic steps; \cite{rahili2017distributed,fazlyab2016self, cherukuri2017saddle} successfully approach dynamic problem in continuous time, but only for isolated systems where time-varying exogenous inputs are available at a central processor.  

In recent years, an intensive research has focused on real-time implementation: \cite{rahili2017distributed} presents a control algorithm for real-time multi-agent systems with the ability to track optimal trajectory, however, only the cost function is time-varying; \cite{tang2017real} proposes an online algorithm for optimal power flow problem based on quasi-Newton method. It can be shown that proposed algorithm is able to provide suboptimal solution at a fast timescale. The tracking ability hinges on the accurate estimation of second order information. For the same application, \cite{dall2018optimal} leverages dual subgradient method and system feedback (or measurements) to design a tracking algorithm based on a double smoothing strategy. Regularization terms are added in both primal and dual subproblems to prove Q-linear convergence to a neighborhood of optimal solution for each time instance; \cite{bernstein2018online} further extend double smoothing algorithm to more general settings and provide a regret analysis. The authors in~\cite{onlineSaddle} presents a saddle-point method for networked online convex optimization, and provide a regret analysis for problems with a Lipshitz-continuous function. Lastly,~\cite{colombino2019online,zheng2019implicit} considered online optimization methods for output regulation problems in dynamical systems. It is also worth mentioning that prediction-correction methods have been utilized to solve time-varying convex optimization problems; see e.g.,~\cite{simonetto2016class,fazlyab2017prediction} and the recent survey paper~\cite{simonetto2020time}.

This paper proposes the development of an online algorithm for time-varying convex problems based on the  alternating direction method of multipliers (ADMM) method\cite{boyd2011distributed}. Using a quadratic regularization term, ADMM can allow one to deal with nonsmooth terms,  and it exhibits improved convergence properties relative to dual (sub)gradient methods, especially for problems with ill-conditioned dual functions \cite{deng2016global}\cite{hong2017linear}. The paper present a new algorithm that has following characteristics: i) at each step the primal subproblems are solved via proximal gradient descent -- providing favorable scalability to large-scale problems and accommodating non-smooth objectives; ii) a dual perturbation method is utilized, where the dual variables are suitably perturbed at every iteration to gain in convergence rate. Related works along this line include the following: \cite{ling2014decentralized} leverages ADMM to solve a real-time multi-agent problem. But it differs from the present work because  it  considers only consensus constraints (a special case of our general formulation); \cite{cao2019dynamic} considers a dynamic sharing problem, and convergence to a neighborhood is provided under standard assumptions; however, the constraint is also a special case of our formulation. 

We note that relative to existing online primal-dual methods~\cite{onlineSaddle,simonetto2014double,bernstein2018online}, the proposed method can handle non-differentiable costs. We also note that the choice of ADMM as opposed to, e.g. an Arrow-Hurwicz method, is due to the following two reasons: i) problem \eqref{P1} has a two-block structure (i.e., $\bx,\by$) which can be effectively handled by ADMM; and, ii) the objective function has nonsmooth terms, which would nevertheless require modifications of the Arrow-Hurwicz method.


To summarize, this paper has the following main contributions: 
\begin{itemize}
    \item We develop an online proximal-ADMM algorithm for solving time-varying optimization problems of  form~\eqref{P1}; 
    \item We provide  convergence analysis, which shows that the proposed algorithm can track the optimal solution trajectory of~\eqref{P1} under mild assumptions; in particular, our methodology does not require smoothness of the cost and does not rely on the full-rankness of the constraint matrix  (see Table \ref{tab:conditions} for detailed comparison between the convergence conditions of a few algorithms).
\end{itemize}
\begin{table}
\begin{center}
\begin{tabular}{ |p{1.0cm}||p{1.3cm}|p{1.5cm}|p{1.7cm}| p{1.1cm}| }
 \hline
 & {\bf Strong Convexity} & {\bf Lipschitz Continuity} & {\bf Full Row Rank} & {\bf Optimality} \\
 \hline 
 \hline
 \multirow{4}{4em}{\bf Classic ADMM} & $f^{(k)}$ & $\nabla f^{(k)}$&$\bA^{(k)},(\bB^{(k)})^T$ & \multirow{4}{4em}{Optimal solution} \\
 \cline{2-4}&$f^{(k)}, g^{(k)}$ & $ \nabla f^{(k)} $ & $\bA^{(k)}$ &\\
 \cline{2-4}&$f^{(k)}$ & $\nabla f^{(k)}, \nabla g^{(k)}$&$(\bB^{(k)})^T$ &  \\
 \cline{2-4}&$f^{(k)}, g^{(k)}$ & $\nabla f^{(k)}, \nabla g^{(k)}$&\xmark &\\
 \hline
 {\bf Proposed Algorithm} &   $f^{(k)}, g^{(k)}$  & \xmark   & \xmark & Perturbed solution [cf.~\eqref{AKKT}] \\
 \hline
\end{tabular}
\end{center}
\caption{Trade off between optimality and conditions for linear convergence \cite{deng2016global}. }\label{tab:conditions}
\end{table}

Our previous work \cite{zhang2017dynamic} focuses on ADMM-based online algorithms  to track a solution of a domain-specific linearized AC optimal power flow (OPF) problem in power grids. In this application domain, this work significantly extends \cite{zhang2017dynamic} in the following ways: 

\begin{itemize}
    \item We consider more general linearized OPF formulations, which can be used to  
    deal with, for example, OPF problems in a distributed setting, where the power system is divided into areas~\cite{autonomous_grids}; and, 
    \item A different algorithm which works under  milder conditions and has wider applicability is proposed. 
\end{itemize}

The remainder of paper is organized as follows. Section \RomanNumeralCaps 2 will give the general time-varying problem formulation. Section \RomanNumeralCaps 3 will introduce our online algorithm. Section \RomanNumeralCaps 4 will apply proposed algorithm to two applications, one is in power systems, the other one is route selection. Tracking ability is shown in \RomanNumeralCaps 5 and \RomanNumeralCaps 6 in the form of convergence analysis and simulation, respectively.

\section{Problem Formulation}

Consider the time-varying  problem \eqref{P1}\footnote{Throughout this paper, boldface characters denote vectors or matrices; characters with superscript ${(k)}$ denote time varying iterates and parameters; for a given vector $\bx$ and matrix $\bG$, $\|\bx\|^2_{\bG} := \bx^T\bG\bx$; $<\bx, \by>$ denotes the inner product between the vectors $\bx$ and $\by$. Given a non-differentiable function $h$, the proximal operator is defined as
$\text{prox}_{h}(\bx) = \arg\min_\bz\|\bz-\bx\|^2 + h(\bz).$}.  At time $k$, if problem \eqref{P1} is solved to global optimality, then we say that the {\it perfect tracking} is achieved. However, in many applications~\cite{dall2019optimization} such perfect tracking may not be possible because before the problem at time $k$ is solved, it may have already evolved to a new problem. Specifically, iterative algorithms often involve multiple iterations of computing and communication, and by the time algorithms converge for time $k$, problem parameters such as  $\bA^{(k)},\bB^{(k)}, \bb^{(k)}$ might have already changed. Therefore it is desirable to design algorithms with certain ``tracking ability", which means that the iterates can be continuously steered to stay close to the time-varying optimal solutions.

Let us reformulate problem \eqref{P1} as follows. First, we rewrite the time-varying constraint sets $\mathcal{X}^{(k)},\mathcal{Y}^{(k)}$ into indicator functions in the objective; and then we separate objective into non-differential functions $f_0^{(k)}(\bx), g_0^{(k)}(\by)$ and differential functions $f_1^{(k)}(\bx), g_1^{(k)}(\by)$. At time $k$ we consider the following  time-varying problem:
\begin{subequations}
\label{P2}
\begin{align} 
\textrm{(P2)~~}\min_{\bx\in\mathbb{R}^m,\by\in\mathbb{R}^n}~&f^{(k)}(\bx) + g^{(k)}(\by) \\
\text{s.t.}~&\bA^{(k)}\bx+\bB^{(k)}\by=\bb^{(k)}\label{P2C}
\end{align}
\end{subequations}
where
\begin{align}
    & f^{(k)}(\bx) \hspace{-1mm}:=\hspace{-1mm} f_0^{(k)}(\bx)+f_1^{(k)}(\bx), g^{(k)}(\by)\hspace{-1mm} := \hspace{-1mm} g_0^{(k)}(\by)+g_1^{(k)}(\by),\nonumber\\
    & \bA^{(k)} := \bA(t_k),\bB^{(k)} := \bB(t_k),\bb^{(k)} := \bb(t_k),\nonumber\\ 
    & \mathcal{X}^{(k)}:=\mathcal{X}(t_k),\mathcal{Y}^{(k)}:=\mathcal{Y}(t_k).\nonumber
\end{align}
Throughout the paper we will assume that the following assumption holds.
\begin{assumption}
\label{Assum_uniform}
For each time $k$, $f^{(k)}, g^{(k)}$ satisfy
\vspace{-1mm}
\begin{align}
\hspace{-4mm}\langle\partial f^{(k)}(\bx_1)\hspace{-1mm}-\hspace{-1mm}\partial f^{(k)}(\bx_2), \bx_1\hspace{-1mm}-\hspace{-1mm}\bx_2\rangle&\hspace{-1mm}\geq\tilde{v}_f\|\bx_1\hspace{-1mm}-\hspace{-1mm}\bx_2\|, \forall \bx_1,\bx_2\label{eqn:convexity_f}\\
\hspace{-3mm}\langle\partial g^{(k)}(\by_1)\hspace{-1mm}-\hspace{-1mm}\partial g^{(k)}(\by_2), \by_1\hspace{-1mm}-\hspace{-1mm}\by_2\rangle&\geq\tilde{v}_g\|\by_1\hspace{-1mm}-\hspace{-1mm}\by_2\|,\forall \by_1,\by_2\label{eqn:convexity_g}
\end{align}
where $\tilde{v}_f,\tilde{v}_g$ are uniform lower bounds of strongly convex constants for $f^{(k)}, g^{(k)}$. Functions $f_1^{(k)}, g_1^{(k)}$ have Lipschitz-continuous gradients, i.e.,
\vspace{-1mm}
\begin{align}
    \|\nabla f_1^{(k)}(\bx_1)\hspace{-1mm} -\hspace{-1mm} \nabla f_1^{(k)}(\bx_2)\|&\leq \tilde{L}_f\|\bx_1\hspace{-1mm}-\hspace{-1mm}\bx_2\|,\forall \bx_1,\bx_2 \label{eqn:Lp_f}\\
    \|\nabla g_1^{(k)}(\by_1) \hspace{-1mm}- \hspace{-1mm}\nabla g_1^{(k)}(\by_2)\|&\leq\tilde{L}_g\|\by_1\hspace{-1mm}-\hspace{-1mm}\by_2\|,\forall\by_1,\by_2\label{eqn:Lp_g}
\end{align}
where $\tilde{L}_f,\tilde{L}_g$ are uniform upper bounds of Lipschitz constants for $\nabla f_1^{(k)}, \nabla g_1^{(k)}$.
\end{assumption}
\begin{assumption}\label{Assum_slater}
For each time $k$, the functions $f^{(k)}(\bx), g^{(k)}(\by)$ are coercive;  i.e.,
\begin{align*}
    f^{(k)}(\bx)\rightarrow\infty~\text{as}~\|\bx\|\rightarrow\infty,~g^{(k)}(\by)\rightarrow\infty~\text{as}~\|\by\|\rightarrow\infty \, .
\end{align*}
\end{assumption}
Assumption \ref{Assum_slater} will be instrumental to ensure that the iterates are bounded. Since continuous coercive functions' level sets $\{\bx| f(\bx) \leq \mu_1, \forall \mu_1\}, \{\by| g(\by) \leq \mu_2, \forall \mu_2\} $ are always compact, the optimal solutions to problem \eqref{P2}, defined as $\bx^{\text{opt},(k)}, \by^{\text{opt},(k)}$, are bounded, i.e.,
\begin{align*}
    &\|\bx^{\text{opt},(k)}\|\leq \sigma_1, ~\|\by^{\text{opt},(k)}\|\leq \sigma_2
\end{align*}
for some positive constants  $\sigma_1, \sigma_2$.

\section{Online Proximal-ADMM using Perturbations}

This section presents an ADMM-based algorithm to track an optimal solution trajectory of the time-varying problem \eqref{P1}. As summarized in Table I, the proposed algorithm exhibits linear convergence guarantees under less stringent  conditions relative to existing ADMM-based methods (even for static problems). In fact, although classic ADMM is conceptually simple and easy to implement, the conditions under which it is convergent is shown to be quite restrictive \cite{deng2016global}. We propose a new algorithm by leveraging the idea of dual perturbation \cite{hajinezhad2019perturbed,koshal2011multiuser} and gradient steps; this will provide a way to demonstrate convergence for a larger family of problems. However, a linear convergence rate at milder conditions comes at the cost of ensuring tracking of an approximate Karush-Kuhn-Tucker (KKT) point~\cite{hajinezhad2019perturbed,koshal2011multiuser,bernstein2018online}. 

Accordingly, we propose to add a small perturbation to the dual variable $\bflambda$ in the form of $1-\beta\gamma$, where $\gamma>0$ is the perturbation parameter and $\beta \gamma \in(0,1)$. The perturbed augmented Lagrangian function is then defined as 
\begin{align}
&\mathcal{L}^{(k)}(\bx,\by;\bflambda) =\mathcal{L}_1^{(k)}(\bx,\by;\bflambda)+  f_0^{(k)}(\bx)+g_0^{(k)}(\by) \label{AL2}
\end{align}
where
\begin{align}
&\mathcal{L}_1^{(k)}(\bx,\by;\bflambda)\hspace{-1mm} = \hspace{-1mm}f^{(k)}_1(\bx)\hspace{-1mm}+\hspace{-1mm}g^{(k)}_1(\by) +\frac{\beta}{2}\|\bA^{(k)}\bx\hspace{-1mm}+\hspace{-1mm}\bB^{(k)}\by\hspace{-1mm}-\hspace{-1mm}\bb^{(k)}\|^2\nonumber\\
&- (1-\beta\gamma)\bflambda^T\hspace{-1mm}(\bA^{(k)}\bx+\bB^{(k)}\by\hspace{-1mm}-\hspace{-1mm}\bb^{(k)}) \nonumber.
\end{align}
Mirroring~\cite{chambolle2011first},  to update $\bx$ and $\by$, one can performs the following steps in an online fashion (where $k$ is here the time index):
\begin{align*}
\by^{(k+1)} \hspace{-1mm}&=\hspace{-1mm}\arg\min_{\by}\;\left\langle\frac{\partial\mathcal{L}_1^{(k+1)}(\bx^{(k)},\by^{(k)};\bflambda^{(k)})}{\partial\by^{(k)}},\by-\by^{(k)}\right\rangle\\
&+ g_0^{(k+1)}(\by)  + \frac{1}{2\alpha_2}\|\by-\by^{(k)}\|^2,\\
\bx^{(k+1)} \hspace{-1mm}&=\hspace{-1mm}\arg\min_{\bx}\; \left\langle\frac{\partial\mathcal{L}_1^{(k+1)}(\bx^{(k)},\by^{(k+1)};\bflambda^{(k)})}{\partial\bx^{(k)}},\bx-\bx^{(k)}\right\rangle\\
&+ f_0^{(k+1)} (\bx) + \frac{1}{2\alpha_1}\|\bx-\bx^{(k)}\|^2,
\end{align*}
where $\alpha_1,\alpha_2$ are step sizes. We are now ready to outline the \emph{online  proximal-ADMM algorithm} with perturbations, whose steps are the following (we stress again that $k$ is  the time index):
\begin{subequations}
\label{perturb}
\begin{align}
\hspace{-1cm}&\hspace{-2mm}\by^{(k+1)} \hspace{-1mm}= \hspace{-1mm}\underset{g_0^{(k+1)}}{\textrm{prox}}\hspace{-1mm}\left(\by^{(k)} \hspace{-1mm}-\hspace{-1mm} \alpha_2\frac{\partial\mathcal{L}_1^{(k+1)}(\bx^{(k)},\by^{(k)};\bflambda^{(k)})}{\partial \by^{(k)}}\right),\label{itr:y}\\
&\hspace{-2mm}\bx^{(k+1)} \hspace{-2mm}= \hspace{-1mm}\underset{f_0^{(k+1)}}{\textrm{prox}}\hspace{-1mm}\left(\bx^{(k)} \hspace{-2mm}- \hspace{-1mm}\alpha_1\frac{\partial\mathcal{L}_1^{(k+1)}(\bx^{(k)},\by^{(k+1)};\bflambda^{(k)})}{\partial \bx^{(k)}}\hspace{-1mm}\right),\label{itr:x}\\
&\hspace{-2mm}\bflambda\hspace{-0.5mm}^{(k+1)} \hspace{-2mm}= \hspace{-1mm}(1\hspace{-1mm}-\hspace{-1mm}\beta\gamma)\bflambda\hspace{-0.5mm}^{(k)}\hspace{-1mm} - \hspace{-1mm}\beta\hspace{-1mm}\left(\bA\hspace{-0.5mm}^{(k+1)}\bx\hspace{-0.5mm}^{(k+1)}\hspace{-1.5mm}+\hspace{-1mm}\bB\hspace{-0.5mm}^{(k+1)}\by\hspace{-0.5mm}^{(k+1)}\hspace{-1.5mm}-\hspace{-1mm}\bb\hspace{-0.5mm}^{(k+1)}\hspace{-1mm}\right)\label{perturb_dual}.
\end{align}
\end{subequations}
The underlying assumption here is that for each time $k$, the interval $\tau$ is sufficient  to run at least one iteration of \eqref{perturb}. 

Compared to classical ADMM-based algorithms (for both static and time-varying optimization),  key differences here are in the proximal gradient steps in the primal update and the perturbation added to $\blambda$. The proximal gradient steps may provide favorable computational gains when applied to a large-scale problem; it also facilitate ones to develop measurement-based algorithms as discussed in, e.g.,~\cite{dall2019optimization}. The perturbation added to $\blambda$ emerges when considering a regularized Lagrangian function of the form $\mathcal{L}^{(k)}(\bx,\by;\blambda) - \frac{\gamma}{2}\|\blambda\|^2$ as in e.g.,~\cite{bernstein2018online,koshal2011multiuser,hajinezhad2019perturbed}. This additional term renders the regularized Lagrangian strongly concave in $\blambda$. Adding a (small) perturbation in dual variable (or, equivalently, considering a regularized Lagrangian) is a very useful technique to ensure convergence of the ADMM, and even obtain a linear convergence behavior as explained shortly. To gain   intuition,  let us consider a toy example as follows:
\begin{align}
\min_{\bx} \quad 0, \quad \mbox{s.t.} \quad \bA \bx =0
\end{align}
where $\bA$ is some fixed matrix, not necessarily positive semidefinite.  
The optimality condition for the above problem can be written down as the following saddle point problem
\begin{align}
\min_{\bx}\max_{\bflambda}~ \bx^T\bA\bflambda\label{perturb_prob}.
\end{align}
One can apply the alternating gradient descent/ascent method for solving problem \eqref{perturb_prob}, whose steps are similar as \eqref{perturb} and are given below 
\begin{align*}
\bx^{(k+1)} &= \bx^{(k)} - \alpha(\bA\bflambda^{(k)}),\\
\bflambda^{(k+1)} &= \bflambda^{(k)} + \beta(\bA^T\bx^{(k)}).
\end{align*}
In Figure \ref{fig:perturb}, we plot $\bx^T\bA\blambda$ using a random matrix $\bA$. An interesting observation  is that the algorithm will diverge if no perturbation is added to $\by$ as shown in Figure \ref{c1}; also see \cite{daskalakis2018limit} for a formal proof. However, once a small perturbation is added to $\by$ in both primal and dual updates, i.e.
\begin{align*}
\bx^{(k+1)} &= \bx^{(k)} - \alpha(\bA\blambda^{(k)}(1-\gamma\beta)),\\
\bflambda^{(k+1)} &= \bflambda^k(1-\gamma\beta) + \beta(\bA^T\bx^{(k)}),
\end{align*}
where $\gamma>0$ is a small number, the algorithm will converge as shown in Figure \ref{c2}. This example serves as a motivation to use the perturbation technique.

\begin{figure}
        \centering
        \begin{subfigure}[b]{0.40\textwidth}
            \centering
            \includegraphics[width=\textwidth]{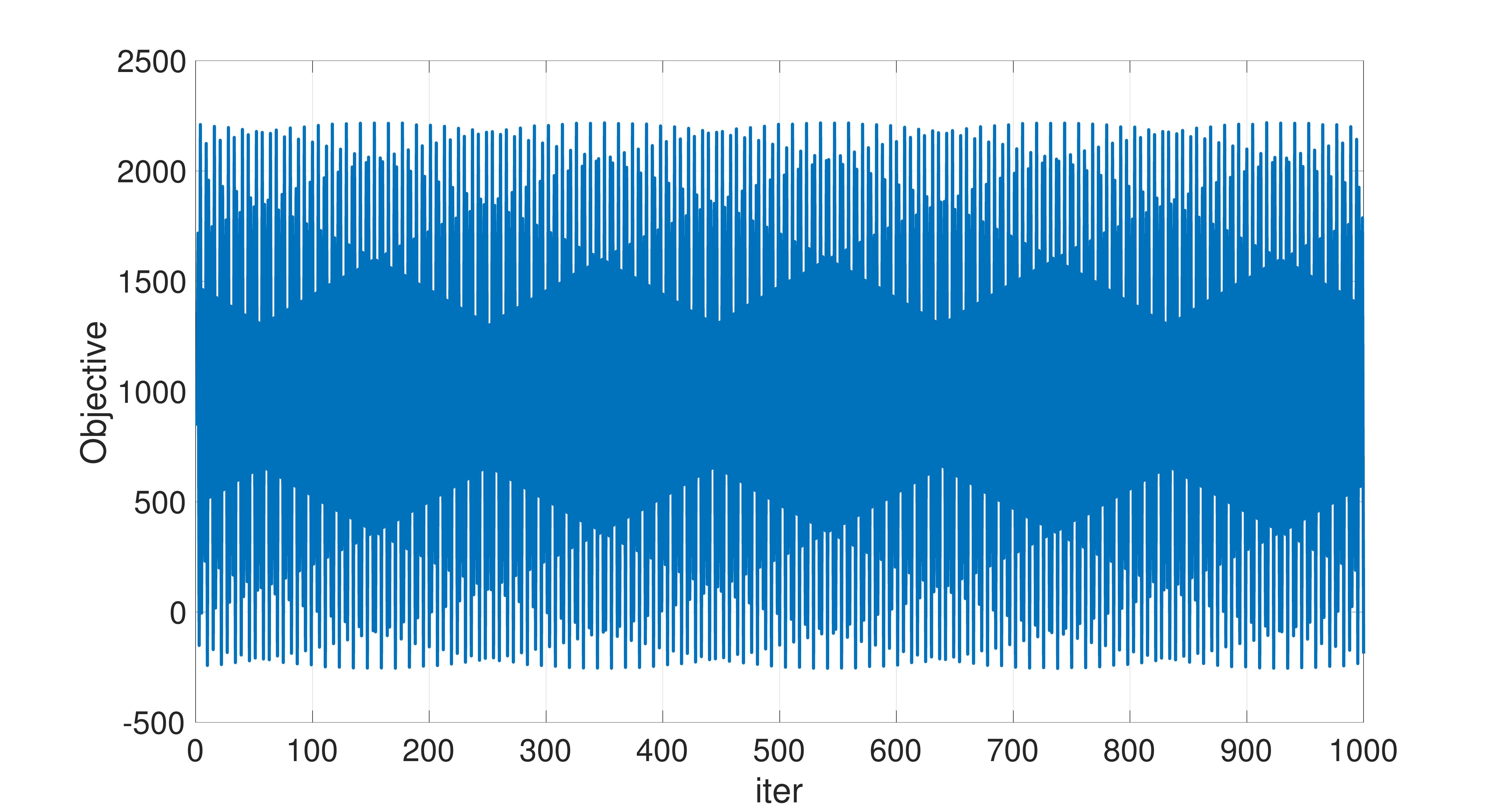}
            \caption[]%
            {{\small Performance without perturbation}}
            \label{c1}
        \end{subfigure}
        \hfill
        \begin{subfigure}[b]{0.40\textwidth}  
            \centering 
            \includegraphics[width=\textwidth]{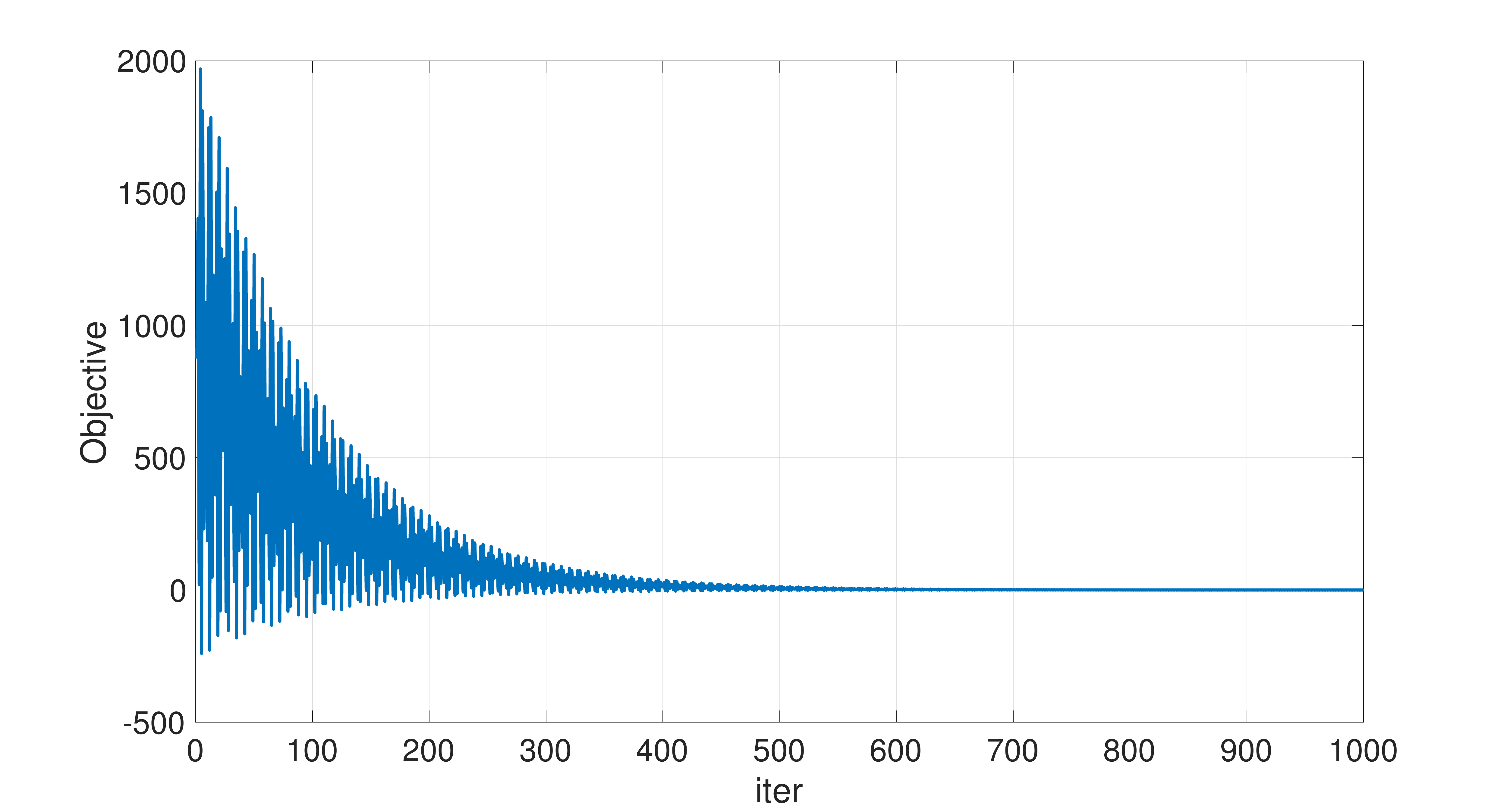}
            \caption[]%
            {{\small Performance with perturbation.}}    
            \label{c2}
        \end{subfigure}
        {\small } 
        \caption{Example of trends of the objective value of \eqref{perturb_prob} for methods with and without perturbation.}
        \label{fig:perturb}
\end{figure}

\section{Convergence Analysis}

In this section we provide analytical results for convergence and tracking ability of the proposed algorithm.

From \cite{deng2016global}, it is known that existing ADMM has relatively strict conditions for linear convergence and these conditions may not hold true in some applications; for example, the coefficient matrices in constraints 
\eqref{P2C} might not have full row rank (this is the case for the application presented later in the paper). Further, in some applications, the objective function of \eqref{P2} may also contain non-smooth terms, which can jeopardize the Lipschitz continuity property.
In contrast, the proposed algorithm could be utilized in a wider range of time-varying optimization problems. It is also worth pointing out that \cite{deng2016global} deals with static optimization problems; here, the focus is on time-varying settings. We begin by first making the following assumption.
\begin{assumption}
\label{Assum1}
For each time $k$, there exists a saddle point $\bw^{\text{opt},(k)} = (\bx^{\text{opt,(k)}},\by^{\text{opt,(k)}},\blambda^{\text{opt,(k)}})$ to problem \eqref{P2} that satisfies the KKT condition:
\begin{align*}
    (\bB^{(k)})^T\blambda^{\text{opt},(k)} &\in \partial g^{(k)}(\by^{\text{opt},(k)}),\\
    (\bA^{(k)})^T\blambda^{\text{opt},(k)} &\in \partial f^{(k)}(\bx^{\text{opt},(k)}),\\
    \bA^{(k)}\bx^{\text{opt},(k)}+&\bB^{(k)}\by^{\text{opt},(k)} = \bb^{(k)},
\end{align*}
where $\blambda^{\text{opt,(k)}}$ is dual variable associated with \eqref{P2C}. The optimal dual variable has a uniform bound, i.e. $\|\blambda^{\text{opt},(k)}\|\leq \mathcal{M}$, where $\mathcal{M}$ is a constant.
\end{assumption}
Assumption \ref{Assum1} is a standard assumption to for convergence analysis \cite{deng2016global}. If Assumption \ref{Assum1} does not hold at a time  $k$, the problem formulation would not be well posed  since there is no solution trajectory to track.  \\
Next,  we analyze the convergence of the algorithm. To proceed, we concatenate primal and dual optimizer as $\{\bw^*\} = \{\bx^*;\by^*;\bflambda^*\}$ (for static case) as the optimizer of $\max\limits_{\blambda}\min\limits_{\bx,\by}\mathcal{L}^{(k)}$ at time $k$. For notation simplicity we neglect superscript $k$ for static case and we have:
\begin{subequations}\label{AKKT}
\begin{align}
\bA^T\bflambda^*-\nabla f_1(\bx^*)\in\partial f_0(\bx^*)\label{x_opt}\\
\bB^T\bflambda^*- \nabla g_1(\by^*) \in \partial g_0(\by^*) \label{y_opt}\\
\bA\bx^*+\bB\by^*-\bb+\gamma\bflambda^* = 0\label{dual_opt}.
\end{align}
\end{subequations}
Condition \eqref{AKKT} is a perturbed version of KKT conditions, related to approximate KKT (AKKT)\cite{dutta2013approximate,andreani2011sequential}. Basically, optimizer $\bw^*$ is not necessarily the KKT point of original problem \eqref{P1}, but rather an approximate solution. 
Let $\{\bx^{\text{opt},(k)}, \by^{\text{opt},(k)}, \blambda^{\text{opt},(k)}\}$ be a KKT point of problem \eqref{P1}, which is also the solution to \eqref{AKKT} when $\gamma = 0$. From Assumption \ref{Assum1}, we know that any optimal dual solution $\blambda^{\text{opt},(k)}$ is bounded. Let $\bv^{\text{opt},(k)} = \{\bx^{\text{opt},(k)}, \by^{\text{opt},(k)}\}, \bv^{*,(k)} = \{\bx^{*,(k)}, \by^{*,(k)}\}$; we can then directly use the result in \cite[Proposition 3.1]{koshal2011multiuser} to show that the distance between $\bv^{*,(k)}$ and $\bv^{\text{opt},(k)}$ is bounded. that is,\footnote{The original result is for smooth strongly convex function, but it is easy to check that it still holds true to our problem.}, 
\begin{align}\label{nedic}
    \kappa\|\bv^{*,(k)} - \bv^{\text{opt},(k)}\|^2 + \frac{\gamma}{2}\|\blambda^{*,(k)}\|^2 \leq \frac{\gamma}{2}\|\blambda^{\text{opt},(k)}\|^2,
\end{align}
where $\kappa>0$ is a constant. We can further derive that 
\begin{align}\label{nedic2}
    \|\bv^{*,(k)} - \bv^{\text{opt},(k)} \|\leq \sqrt{\gamma}\max\|\blambda^{\text{opt},(k)}\|\cdot\textrm{c}, \; \forall~k
\end{align}
where $c>0$ is some constant; see \cite[Proposition 3.1]{koshal2011multiuser} for a detailed discussion and the  proof of the result. Since every optimal dual variable $\blambda^{\text{opt},(k)}$ to the original problem \eqref{P1} is bounded, i.e. $\|\blambda^{\text{opt},(k)}\|\leq \mathcal{M}$, it follows that the distance between the AKKT point and a KKT point is bounded too, and such a bound depends on the choice of $\gamma$ (the smaller the $\gamma$, the smaller the distance since $\|\blambda^{\text{opt},(k)}\|$ is independent of $\gamma$). 
In the following analysis, we focus on bounding the distance between iterates generated by \eqref{perturb} and AKKT points.
To proceed, we first state the assumptions on $\bx^{*,(k)}, \by^{*,(k)}$ and on problem parameters as follows.
\begin{assumption}
\label{Assum_Bounded}
The successive difference between AKKT points is bounded:
\begin{align}
\|\bx^{*,(k+1)} - &\bx^{*,(k)}\|\leq \sigma_{\bx}, \,\,\,\, \|\by^{*,(k+1)} - \by^{*,(k)}\|\leq \sigma_{\by},\label{diff}
\end{align}  
where $\bx^{*,(k)},\by^{*,(k)}$ are primal optimizer of \eqref{AKKT} at time $k$; $\sigma_\bx>0, \sigma_\by >0$ are some constants. Also, the variation of the problem parameters is bounded as:
\begin{align}
    &\|\bA^{(k+1)} - \bA^{(k)}\| \leq \sigma_\bA, ~\|\bB^{(k+1)} - \bB^{(k)}\| \leq \sigma_\bB \label{diff_ab}\\
    &\|\bA^{(k)}\|\leq\tilde{\sigma}_{\bA},~\|\bB^{(k)}\|\leq\tilde{\sigma}_{\bB},~\|\bb^{(k+1)}-\bb^{(k)}\|\leq\sigma_\bb\label{parameter}
\end{align}
where $\sigma_\bA,\sigma_\bB,\tilde{\sigma}_{\bA},\tilde{\sigma}_{\bB},\sigma_\bb$ are some given positive constants.
\end{assumption}
Assumption \ref{Assum_Bounded} is  common in time-varying optimization~\cite{dall2019optimization, simonetto2014double, fazlyab2016self,bernstein2018online,Shahrampour2018,hall2015online,madden2020bounds}; worst-case bounds for \eqref{diff} can be obtained assuming that  the  sets $\mathcal{X}^{(k)},\mathcal{Y}^{(k)}$ are compact uniformly in time. Another approach is to measure the  distance based on the \textit{optimal drift}, without assuming a specific bound; see, e.g.,~\cite{ling2014decentralized,cao2019dynamic}. The parameters  $\sigma_{\bx}$ and $\sigma_{\by}$ quantify the maximum variation of the optimal solutions over two consecutive time steps. Since the paper deals with a tracking problem, conventional wisdom would suggest that better tracking performance can be achieved when  \eqref{P1} is not changing rapidly; this will be confirmed in the convergence results presented later (see also ~\cite{dall2019optimization}).

The main result of the paper is stated next. For notation simplicity, let $\{\bw^{(k)}\}=\{\bx^{(k)};\by^{(k)};\blambda^{(k)}\}$ be the iterates generated by~\eqref{perturb}, and  let  $\{\bw^{*,(k)}\}=\{\bx^{*,(k)},\by^{*,(k)},\blambda^{*,(k)}\}$ be an optimizer of $\max\limits_{\blambda}\min\limits_{\bx,\by}\mathcal{L}^{(k)}$.

\begin{theorem}\label{main}
Suppose that Assumptions~\ref{Assum1}--\ref{Assum_Bounded} hold for each time $k$.  Let $\bG =
\textrm{diag}(\frac{1}{\alpha_1}\bI, \frac{1}{\alpha_2}\bI, \frac{1}{\beta}\bI)$ be a  positive definite matrix. Assume that the  step size $\beta$ and the perturbation constant $\gamma$ satisfy: $\beta\gamma + \beta \leq 1, \beta \leq 1$. Finally, assume that the step sizes satisfy the  following:
\begin{align*}
    0 & < \alpha_1 \leq \left((1+\beta\gamma)\tilde{\sigma}^2_\bA + \frac{\tilde{L}_f^2}{\tilde{v}_f} \right)^{-1}  \nonumber \\
    0 & < \alpha_2 \leq \left(\frac{2\beta^2\max\tilde{\sigma}^4_\bB}{\tilde{v}_g}+ \frac{\tilde{L}_g^2}{\tilde{v}_g} \right)^{-1} \, .
\end{align*}

Then, at every time $k$, the tracking error of the algorithm~\eqref{perturb} evolves as:  
\begin{align}
	\|\bw^{(k)} \hspace{-1mm}-\hspace{-1mm} \bw^{*,(k)}\|_\bG & \leq 
	 \frac{1}{1+\delta} \|\bw^{(k-1)}-\bw^{*,(k-1)}\|_\bG \nonumber \\
	 & \,\, + \frac{1}{1+\delta} \left(\frac{\sigma_\bx^2}{\alpha_1} +\frac{\sigma_\by^2}{\alpha_2} + 2\sigma_{\blambda}^2\right)^{\frac{1}{2}}
	 \label{eq:linear_conv}
	\end{align}
where $ \sigma_{\blambda} :=\tilde{\sigma}_{\bA}\sigma_\bx+\tilde{\sigma}_{\bB}\sigma_\by+\sigma_\bb+\sigma_\bA\mathcal{J}(\sigma_1) + \sigma_\bB\mathcal{J}(\sigma_2)$, $\mathcal{J}(\sigma_1) = \sigma_1 + \sqrt{\gamma}\mathcal{M}c, \mathcal{J}(\sigma_2) = \sigma_2 + \sqrt{\gamma}\mathcal{M}c$ and  $\delta$ satisfies the condition:  
\begin{align}
0 < \delta \leq \min\left(\frac{\tilde{v}_f}{(1+\beta\gamma)\tilde{\sigma}^2_\bA + \frac{\tilde{L}_f^2}{\tilde{v}_f}},\frac{\tilde{v}^2_g}{4\beta^2\tilde{\sigma}^4_\bB + 2\tilde{L}_g^2}, \beta\gamma\right).
\end{align}
\end{theorem} 

\begin{corollary}
\label{linear}
Under the assumptions of Theorem~\ref{main}, one has the following asymptotic behavior for the tracking error:   
\begin{align}
	\label{eq:result_theorem}
		\limsup_{k\rightarrow\infty} \|\bw^{(k)} - \bw^{*,(k)}\|_\bG \leq 
\frac{1}{\delta} \left(\frac{\sigma_\bx^2}{\alpha_1} +\frac{\sigma_\by^2}{\alpha_2} + 2\sigma_{\blambda}^2\right)^{\frac{1}{2}} .
\end{align}
\end{corollary}

The proofs of Theorem~\ref{main} and Corollary~\ref{linear} are provided in the Appendix.\\
\textit{Remark.} We also provide a comment on the boundedness of the dual iterates; specifically, $\blambda^{(k)}$ can be bounded as follows:
\begin{align*}
    \|\blambda^{(k)}\| &= \|\blambda^{(k)} - \blambda^{*,(k)} + \blambda^{*,(k)}\| \\&\leq \|\blambda^{(k)} - \blambda^{*,(k)}\| + \|\blambda^{*,(k)}\| \, .
\end{align*}
From Corollary 1, one has that $\|\blambda^{(k)} - \blambda^{*,(k)}\|$ is bounded; further, from \cite[Proposition 3.1]{koshal2011multiuser} and \eqref{nedic} it follows that $\|\blambda^{*,(k)}\|\leq \|\blambda^{\text{opt},(k)}\|$, where $\blambda^{\text{opt},(k)}$ is the optimal dual solution (without perturbation), which is assumed to be bounded as $\|\blambda^{\text{opt},(k)}\|\leq\mathcal{M}$. Therefore, $\|\blambda^{(k)}\|$ is bounded too. The distance between $\bw^{(k)}$ and the optimal solution $\bw^{\text{opt},(k)} = \{\bx^{\text{opt},(k)},\by^{\text{opt},(k)},\blambda^{\text{opt},(k)}\}$ can be bounded as follows:
\begin{align*}
    &\|\bw^{(k)} - \bw^{\text{opt},(k)}\|_\bG = \|\bw^{(k)} - \bw^{*,(k)} \hspace{-1mm}+\hspace{-1mm} \bw^{*,(k)}\hspace{-1mm} -\hspace{-1mm} \bw^{\text{opt},(k)}\|_\bG \\
   &\leq \|\bw^{(k)} - \bw^{*,(k)}\|_\bG + \|\bw^{*,(k)} - \bw^{\text{opt},(k)}\|_\bG\\
   &\stackrel{\eqref{eq:result_theorem}}{\leq}\frac{1}{\delta} \left(\frac{\sigma_\bx^2}{\alpha_1} +\frac{\sigma_\by^2}{\alpha_2} + 2\sigma_{\blambda}^2\right)^{\frac{1}{2}} + \|\bw^{*,(k)} - \bw^{\text{opt},(k)}\|_\bG.
\end{align*}
From \eqref{nedic2} and $\|\blambda^{\text{opt},(k)}\|\leq\mathcal{M}$, we know that 
\begin{align*}
    &\|\bw^{*,(k)} - \bw^{\text{opt},(k)}\|^2_\bG\\
    =&\|\bv^{*,(k)} - \bv^{\text{opt},(k)}\|^2_{(\frac{1}{\alpha_1};\frac{1}{\alpha_2})}+\|\blambda^{*,(k)} - \blambda^{\text{opt},(k)}\|^2_{\frac{1}{\beta}}\\
    \leq&\max(\frac{1}{\alpha_1}, \frac{1}{\alpha_2})\gamma\mathcal{M}^2+\frac{2}{\beta}\mathcal{M}^2.
\end{align*}
This eventually gives us
\begin{align*}
    &\|\bw^{(k)} - \bw^{\text{opt},(k)}\|^2_\bG \\
    \leq& \frac{2}{\delta^2} \left(\frac{\sigma_\bx^2}{\alpha_1} +\frac{\sigma_\by^2}{\alpha_2} + 2\sigma_{\blambda}^2\right) + 2 {\max(\frac{1}{\alpha_1}, \frac{1}{\alpha_2})\gamma\mathcal{M}^2+\frac{4}{\beta}\mathcal{M}^2} \, .
\end{align*}
The result~\eqref{eq:linear_conv} asserts that the proposed algorithm exhibits linear convergence with a contraction coefficient of $1/(1-\delta)$~\cite{dall2019optimization}; the evolution of the tracking error depends on the temporal variability of the optimal trajectory, which is bounded by the second term on the right-hand-side of~\eqref{eq:linear_conv}. It is worth pointing out that, if the problem~\eqref{P1} is static, then~\eqref{eq:linear_conv} boils down to   
$$
\|\bw^{(k)} \hspace{-1mm}-\hspace{-1mm} \bw^{*}\|_\bG \leq 
	 \frac{1}{1+\delta}\|\bw^{(k-1)}-\bw^{*}\|_\bG
$$
showing linear convergence of the proximal-ADMM method in batch optimization ($\bw^{*}$ is in this case the solution of the static problem).

The asymptotic result~\eqref{eq:result_theorem} matches existing results in online methods for time-varying optimization~\cite{bernstein2018online,dall2019optimization,dixit2019online}. In particular, the bound depends on the $\delta$ (which affects the contraction coefficient) and the maximum variation of the optimal trajectory over two consecutive time steps, and it shows  how the variation of the problem parameters and optimal solutions can affect the  tracking performance.

Although~\eqref{eq:result_theorem} asserts that the maximum tracking error is bounded, its tightness is to be investigated on a case-by-case basis (i.e., based on the particular evolution of the solution)\footnote{For example, ~\cite{madden2020bounds} showed that a bound of the form~\eqref{eq:result_theorem} is actually tightly met for online gradient and proximal-gradient descent  for a particular sequence of adversarial cost functions.}.

The following corollary presented  shows how to maximize $\delta$ (and, hence, how to minimize the worst-case tracking bound). 

\begin{corollary}
\label{stepsize}
If the step sizes are selected as 
\begin{align*}
    \alpha_1 =\hspace{-1mm} \left(\hspace{-1mm}(1+\beta\gamma)\tilde{\sigma}^2_\bA + \frac{\tilde{L}_f^2}{\tilde{v}_f} \right)^{-1}\hspace{-2mm}, 
    \alpha_2 =\hspace{-1mm} \left(\frac{2\beta^2\max\tilde{\sigma}^4_\bB}{\tilde{v}_g}+ \frac{\tilde{L}_g^2}{\tilde{v}_g} \right)^{-1} 
\end{align*}
then one has that the contraction coefficient $\frac{1}{1+ \delta}$ can be computed using the following expression:
\begin{align}
\delta = \min\left(\frac{\tilde{v}_f}{(1+\beta\gamma)\tilde{\sigma}^2_\bA + \frac{\tilde{L}_f^2}{\tilde{v}_f}},\frac{\tilde{v}^2_g}{4\beta^2\tilde{\sigma}^4_\bB + 2\tilde{L}_g^2}, \beta\gamma\right).
\end{align}
\end{corollary}

Notice that, for example, if we specify $\beta=0.5$ and $ \gamma=1$, then $\alpha_1,\alpha_2, \delta$ depend only  on the problem itself; i.e,
\vspace{-2mm}
\begin{align}
    \delta &= \min\left(\frac{\tilde{v}_f}{\frac{3}{2}\tilde{\sigma}^2_\bA + \frac{\tilde{L}_f^2}{\tilde{v}_f}},\frac{\tilde{v}^2_g}{\tilde{\sigma}^4_\bB + 2\tilde{L}_g^2}, \frac{1}{2}\right),\label{eqn:delta1}\\
    \alpha_1 &= \frac{1}{\frac{3}{2}\tilde{\sigma}^2_\bA + \frac{\tilde{L}_f^2}{\tilde{v}_f}}, ~\alpha_2 = \frac{1}{\frac{\max\tilde{\sigma}^4_\bB}{2\tilde{v}_g}+ \frac{\tilde{L}_g^2}{\tilde{v}_g}}.\label{eqn:a12}
\end{align}
As long as one picks $\delta$ as in \eqref{eqn:delta1}, there exist suitable $\alpha_1,\alpha_2$ to ensure convergence (see \eqref{eqn:alpha1}--\eqref{eqn:alpha2} in the proof).

\section{Example of Motivating Applications}\label{simulation}

\subsection{Multi-area power grid optimization}\label{power}

In this section, we briefly outline an example in power grids. We consider a distribution network featuring distributed energy resources (DERs), and we apply the proposed methodology to drive the DER output powers to the solution of an optimization problem encapsulating voltage constraints and given performance objectives. We demonstrate that the proposed methodology is amenable to settings where the distribution system is partitioned in areas; each area is autonomously controlled, and it ``trades'' power with adjacent areas based on given economic objectives~\cite{autonomous_grids}. In contrast,  previous works in the context of real-time optimal power flow involve centralized algorithms~\cite{tang2017real,hauswirth2016projected} or algorithms with a gather-and-broadcast architecture~\cite{dall2018optimal}.     

Similar to~\cite{autonomous_grids}, consider partitioning a power distribution network into $C$ clusters, and denote as $\cC_{i}$ the set of electrical nodes within cluster $i = 1, \ldots, C$. Two clusters $i$ and $j$ are adjacent if there is at least an electrical node $i$ such that $i \in \cC_{i}$ and $i \in \cC_{j}$. Let $\cB_{i,j} := \cC_{i} \cap \cC_{j}$ be the set of \emph{boundary nodes} connecting cluster $i$ to cluster $j$, and define $\cB_{i} := \cup_{j\neq i} \cB_{i,j}$. Further, let $\cI_{i} := \cC_{i} \backslash \cB_{i}$ be the set of \emph{internal nodes} for cluster $i$. For future developments, let $N_{i} := |\cI_{i}|$ be the number of internal nodes if cluster $i$, and let $\cN_{i} \subset \{1, \ldots, C\}$ be the set of neighboring clusters of the $i$th one (i.e., cluster connected to the $i$th one). 

Let $\bx_{j}^{i} := [P_{j}^{i}, Q_{j}^{i}]^\sfT \in \reals^2$ collect the net active and reactive powers injected by DERs at the node $j \in \cI_{i}$ of cluster $i$. Particularly, $\bx_{j}^{i} $ can represent the powers injected by \emph{one} DER located at node $j$, or the aggregate net power injections of \emph{a group of} DERs located at node $j$ (e.g., a household with multiple controllable devices) and we stack the setpoints $\{\bx_{j}^{i}\}_{ j \in \cI_{i}}$ in the vector $\bx^{i} \in \reals^{2 N_{i}}$. If no controllable DERs are present at a given location, the corresponding vector  $\bx_{j}^{i}$ is  set to $\mathbf{0}$\footnote{For notation simplicity, the model is outlined for balanced systems and for the case where one household/building with DERs is located at a node. However, the model can be trivially extended to multiphase networks~\cite{linModels} and for the case where multiple households/buildings with DERs are located at a node (at the cost of increasing the complexity of the notation).}. On the other hand, $\bell_{j}^{i} \in \reals^2$ denotes the net non-controllable loads at node $j \in \cI_{i}$, and $\bell^{i} \in \reals^{2 N_{i}}$ stacks the loads $\{\bell_{j}^{i}\}_{ j \in \cI_{i}}$. It is assumed that no DERs and no non-controllable loads are located at the boundary nodes $\cB_{i,j}$. 

Let $V_{j}^{i} \in \mathbb{C}$ denote the complex line-to-ground voltage phasor at node $j$ of cluster $i$, and let $\bv^{i} := [\{|V_{j}^{i}|, j \in \cI_{i}\}]^\sfT$ be the vector of voltage  magnitudes of the internal nodes $\cI_{i}$. For each pair of neighboring clusters $(i,j)$, let $\bx_{n}^{j \rightarrow i} := [P_{n}^{j \rightarrow i}, Q_{n}^{j \rightarrow i}]^\sfT \in \reals^2$ represent the active and reactive powers flowing into area $i$ from area $j$ through node $n \in \cB_{i,j} $; on the other hand, $\bx_{n}^{i \rightarrow j} \in \reals^2$ contains the active and reactive powers flowing into area $j$ from area $i$ through node $n \in \cB_{i,j} $. From Kirchhoff's Law, it holds that $\bx_{n}^{j \rightarrow i} + \bx_{n}^{i \rightarrow j} = \mathbf{0}$. To facilitate the syntheses of computationally-affordable algorithms, we leverage the following approximate linear relationship between net injected power and voltage magnitude (see e.g.,~\cite{dall2018optimal} and references therein):
\begin{subequations} \label{relation1}
\begin{align}
\tilde{\bv}^{i}  &:=  \hspace{-2mm} \sum_{j \in \cI_{i}} \bA_{j}^{i} (\bx^{i}_j - \bell^{i}_j) +\hspace{-2mm}  \sum_{j \in \cN_{i}} \sum_{n \in \cB_{i,j}}\hspace{-2mm} \bA_{n}^{j \rightarrow i} \bx_{n}^{j \rightarrow i} \hspace{-1mm}+\hspace{-1mm} \ba, \label{eqn:lin_v} \\
 &=  \bA^{i} (\bx^{i} - \bell^{i}) +  \sum_{j \in \cN^{i}}  \bA^{j \rightarrow i} \bx^{j \rightarrow i} + \ba, \label{eqn:lin_v2} 
\end{align}
\end{subequations}
where $\bA^i = [\bA_j^i]_{j\in\cI_i}, \bA^{j\rightarrow i} = [\bA^{j\rightarrow i}_n]_{n\in\cB_{i,j}}, \ba$ are time-varying problem parameters derived from linearized power flow equation. 
Another linear relationship between net injected power and power between clusters is captured in the following equation:
\begin{subequations}\label{relation2}
\begin{align} 
\bx^{j \rightarrow i} &:=  \sum_{k \in \cI_{i}} \bM^{j \rightarrow i}_{k}  (\bx^{i}_k - \bell^{i}_k)  + \mathbf{m}^{j \rightarrow i}, \label{eqn:lin_s0} \nonumber\\  + & \sum_{k \in \cN_{i} \backslash \{j\}} \sum_{n \in \cB_{i,k}} \bM^{k, j \rightarrow i}_{n}\bx_{n}^{k \rightarrow i}\\
&=   \bM^{j \rightarrow i}  (\bx^{i}\hspace{-1mm} - \hspace{-1mm}\bell^{i}) + \mathbf{m}^{j \rightarrow i}+ \hspace{-4mm} \sum_{k \in \cN_{i} \backslash \{j\}} \hspace{-3mm}\bM^{k, j \rightarrow i} \bx^{j \rightarrow i}, \label{eqn:lin_s01} 
\end{align}
\end{subequations}
where $\bM^{j\rightarrow i} \hspace{-1mm}=\hspace{-1mm} [\bM^{j\rightarrow i}_k]_{k\in\cI_i}, \bM^{k,j\rightarrow i} \hspace{-1mm}=\hspace{-1mm} [\bM^{k,j\rightarrow i}_{n}]_{n\in\cB_{i,k}}$,$\bm^{j\rightarrow i}$ are also time-varying problem parameters depending on the actual network physics. 
All model parameters in \eqref{relation1}--\eqref{relation2} can be obtained as shown in~\cite{linModels}. 
Now we are ready to state our real-time OPF problem as follows:
\begin{subequations} 
\begin{align} 
&\hspace{-5mm}\min_{\substack{\{\bx^{i}\}, \{\bx_{n}^{j \rightarrow i}, \bx_{n}^{i \rightarrow j}\}}}  \hspace{.2cm} \sum_{i=1}^C  [f^{i}(\bx^{i}) + g^i(\{\bx^{j \rightarrow i}\})] \tag{P3} \label{obj} \\
 \text{s.t.} ~&\bx_{j}^{i} \in \cY_{j}^{i}, \forall \,\, j \in \cI^{i}, \,\, i = 1, \ldots C \label{eqn:constr_Y_exact}  \\
& v^{\text{min}} \mathbf{1} \leq \tilde{\bv}^{i} \leq v^{\text{max}} \mathbf{1}, \forall \,\, i = 1, \ldots C \label{cons}\\
& \bx^{j \rightarrow i} =   \bM^{j \rightarrow i}  (\bx^{i} - \bell^{i}) + \mathbf{m}^{j \rightarrow i}+\hspace{-5mm}  \sum_{k \in \cN_{i} \backslash \{j\}} \hspace{-5mm}\bM^{k, j \rightarrow i} \bx^{j \rightarrow i} \nonumber\\
&\hspace{3cm}, \, \forall \, j \in \cN_{i}, \, \, i = 1, \ldots, C \label{between}\\
& \bx^{j \rightarrow i} + \bx^{i \rightarrow j} = \mathbf{0} , \,\, \forall \,\, \textrm{neighboring~areas~} (i,j)
\end{align}
\end{subequations}
where the time-varying objective function models the amount of real power curtailed and the amount of reactive power injected or absorbed (which leads to non-smooth term in the objective, e.g. $\ell_1$ term). For notation simplicity, we write objective function in \eqref{obj} as $\Psi(\bx)$. Consider $\bM^{j\rightarrow i}$ consists of $1,0$, with $1$ for real power, $0$ for reactive power. Putting \eqref{eqn:lin_s01} back to \eqref{eqn:lin_v2}, adding slack variables $\bgamma^i,\bbeta^i$ to \eqref{cons} formulate equality constraints, and adding strongly convex term w.r.t $\bgamma = \{\bgamma^i\},\bbeta = \{\bbeta^i\}$ we have the following formulation:
\begin{subequations} 
\begin{align} 
 &\hspace{-0.5cm}\min_{\substack{\{\{\bx^{i}\}, \{\bx_{n}^{j \rightarrow i}, \bx_{n}^{i \rightarrow j}\}\},\{\bgamma^{i},\bbeta^{i}\geq 0\}\} } } \Psi(\bx)+ a\|\bgamma\|^2 + b\|\bbeta\|^2\tag{P4} \label{obj2} \\
 \text{s.t.~} & \bx_{j}^{i} \in \cY_{j}^{i},  \forall \,\, j \in \cI_{i}, \,\, i = 1, \ldots C \label{eqn:constr_Y_exact}  \\
& v^{\text{min}}\mathbf{1} -\tilde{\bv}^{i} + \bgamma^{i}= \mathbf{0},  \forall \,\, i = 1, \ldots C \label{cons2}\\
&\tilde{\bv}^{i} + \bbeta^{i} - v^{\text{max}}\mathbf{1}=\mathbf{0} ,  \forall \,\, i = 1, \ldots C \label{cons3}\\
& \bx^{j \rightarrow i} =   \bM^{j \rightarrow i}  (\bx^{i} - \bell^{i}) + \mathbf{m}^{j \rightarrow i}+\hspace{-2mm}  \sum_{k \in \cN_{i} \backslash \{j\}}\hspace{-2mm} \bM^{k, j \rightarrow i} \bx^{j \rightarrow i} \nonumber\\
&\hspace{3cm}, \, \forall \, j \in \cN_{i}, \, \, i = 1, \ldots, C \label{eqn:between}\\
& \bx^{j \rightarrow i} + \bx^{i \rightarrow j} = \mathbf{0} , \,\, \forall \,\, \textrm{neighboring~areas~} (i,j).\label{eqn:consensus}
 \end{align}
\end{subequations}
We can now clearly see a mapping from \eqref{obj2} to \eqref{P1}: objective functions
are $\Psi(\bx)$ and $a\|\bgamma\|^2 + b\|\bbeta\|^2$(where $a,b > 0$ are small); two blocks of variables are $\{\bx^{i},\bx^{j\rightarrow i}\}$ and $\{\bgamma^{i},\bbeta^{i}\}$; constraints are all linear and separable w.r.t each network node. \textcolor{black}{Problem \eqref{obj2} is time varying in both objective function and constraint parameters.} In order to better illustrate how the proposed algorithm can be applied, we use a 4-cluster network (see Figure \ref{fig:net}) as an example. First, we substitute $\bx^{j\rightarrow i}$ in \eqref{eqn:lin_v2} with \eqref{eqn:between}; then, we substitute $\tilde{\bv}^i$ in \eqref{cons2}--\eqref{cons3} with \eqref{eqn:lin_v2};
last, we define the corresponding augmented Lagrangian function as follows:
\vspace{-2mm}
\begin{align*}
&\mathcal{L}(\bx,\bgamma,\bbeta,\blambda) = \Psi(\bx) + a\|\bgamma\|^2 + b\|\bbeta\|^2 \nonumber\\
+ &\sum_{k}\hspace{-1mm}\sum_{i\in C_k}\frac{\rho}{2}\|\textcolor{black}{\sum_{j\in C_k}(\bA^{i}_j+\bA_j^{3})\bx^{i}_j  \hspace{-1mm}+\hspace{-1mm} \ba_k }\hspace{-1mm}+\hspace{-1mm}\bbeta^{i} \hspace{-1mm}-\hspace{-1mm} v^{\text{max}} \hspace{-1mm}+ \hspace{-1mm}\frac{\blambda_1(1-\gamma)}{\rho}\|^2\\
+&\sum_{k}\hspace{-1mm}\sum_{i\in C_k} \frac{\rho}{2}\|v^{\text{min}} \hspace{-1mm}\textcolor{black}{-\hspace{-1mm}\sum_{j\in C_k} (\bA^{i}_j\hspace{-1mm}+\hspace{-1mm}\bA^{3}_j)\bx^{i}_j \hspace{-1mm}- \hspace{-1mm}\ba_k} \hspace{-1mm}+\hspace{-1mm} \bgamma^{i}\hspace{-1mm}+\hspace{-1mm}\frac{\blambda_2(1-\gamma)}{\rho}\|^2\\
+& \sum_{k}\frac{\rho}{2}\|\sum_{i}\bx^{i\rightarrow j} - \textcolor{black}{\sum_{j\in C_k} \bx^{j}} + \frac{\blambda_4(1-\gamma)}{\rho}\|^2\\
+&\frac{\rho}{2}\|\sum_{i\neq j}\bx^{i\rightarrow j}+\frac{\blambda_3(1-\gamma)}{\rho}\|^2 
\end{align*}

The detailed updates follow the same way as \eqref{perturb} and from Table I we know that linear convergence to AKKT is guaranteed. To further improve our algorithm for this particular application. We incorporate system measurements in both primal and dual updates in the following way:
\begin{align}
\sum_{j\in C_k}(\bA^{i}_j+\bA_j^{3})\bx^{i}_j +\ba_k &\rightarrow \phi(\bx),~ \sum_{j\in C_k}\bx^{j}\rightarrow \psi(\bx),
\end{align}
where $\phi(\bx), \psi(\bx)$ are measurements. This is beneficial in that: i) A natural distributed computing scheme is achieved while without feedback it is not clear whether the algorithm can be implemented in a distributed way; ii) feedback terms are much less than uncontrollable terms, which essentially shrinks the measuring time; iii) it is easier to satisfy power flow equations with the help of system measurements.

\begin{figure}
\centering
\includegraphics[width=8cm]{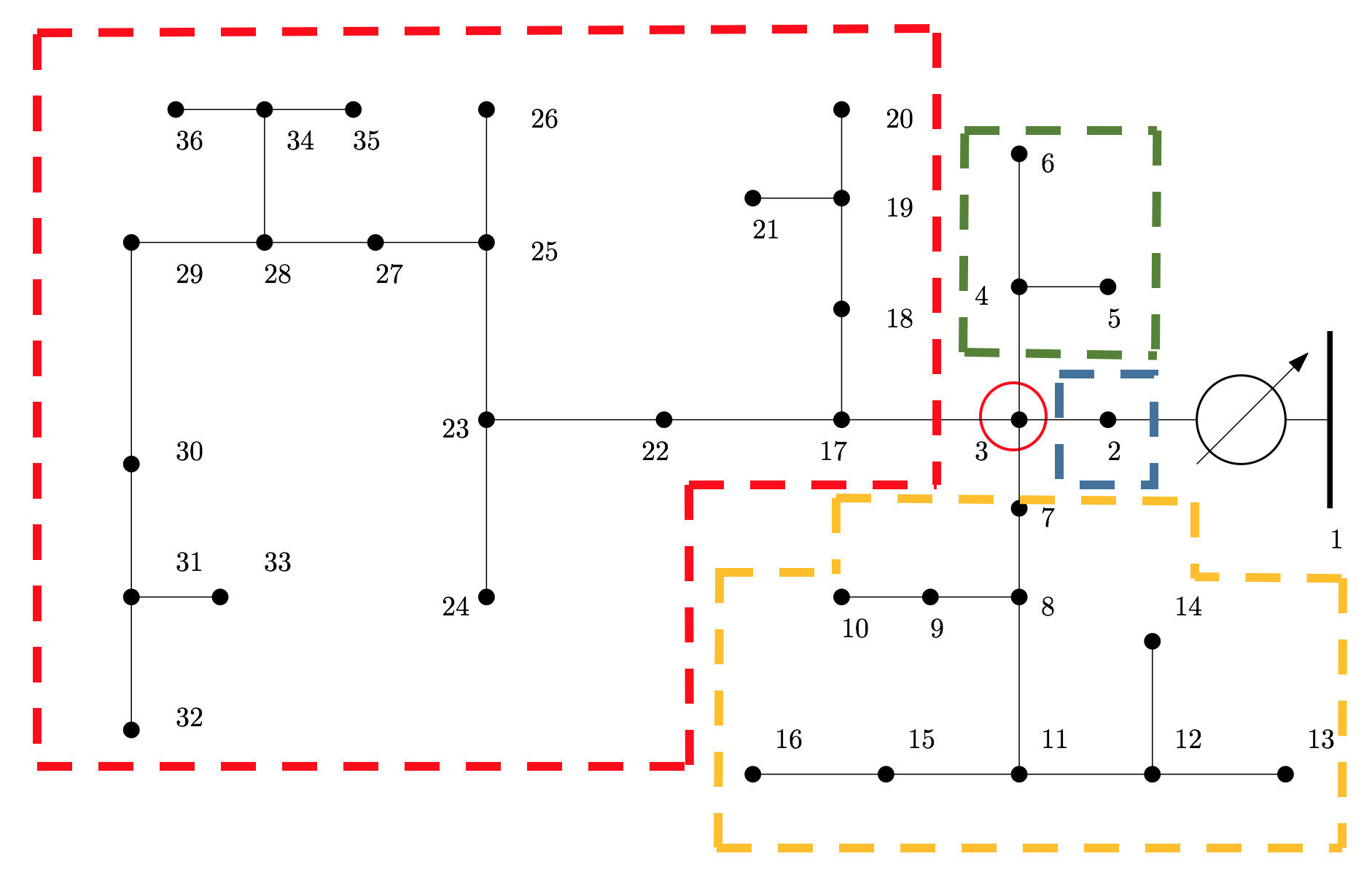} 
\caption{power distribution network with 4 clusters, node 3 is a boundary node that belongs to all 4 clusters.}
\label{fig:net}
\end{figure}

\subsection{Simulations}
In this section, we test our algorithm using the same power systems settings. We consider a similar system as in \cite{dall2018optimal}, where a modified IEEE 37-node test feeder is utilized. The network is obtained by considering a single phase equivalent, and by replacing the loads on phase ``c" specified in the original dataset with real load data measured from feeders in a neighborhood called Anatolia in California during a week in August 2012. It is assumed that the aggregations of photovoltaic systems are located at nodes 4, 7, 10, 13, 17, 20, 22, 23, 26, 28, 29, 30, 31, 32, 33, 34, 35, and 36. The rating of  these inverters are 300kVA for $i=3$, 350kVA for $i=15, 16$ and 200kVA for the remaining ones. The objective is set to be $f^i(\bx^i) = c_p(P_{\text{av},i} - P_i)^2 + c_q(Q_i)^2 + \bar{c}_q|Q_i|, g^i(\bx^{j\rightarrow i}) = 0$ where $P_{\text{av},i}$ is the maximum real power available from the PV system $i$, and $c_p = 3, c_q = 1, \bar{c}_q = 0.1$. The voltage limits are set to be $V^{\text{min}} = 0.95$pu, $V^{\text{max}} = 1.05$pu. The generation profiles are simulated based on real solar irradiance data and have a granularity of 1 second. 
\begin{figure}
	\centering
	\includegraphics[width=8cm]{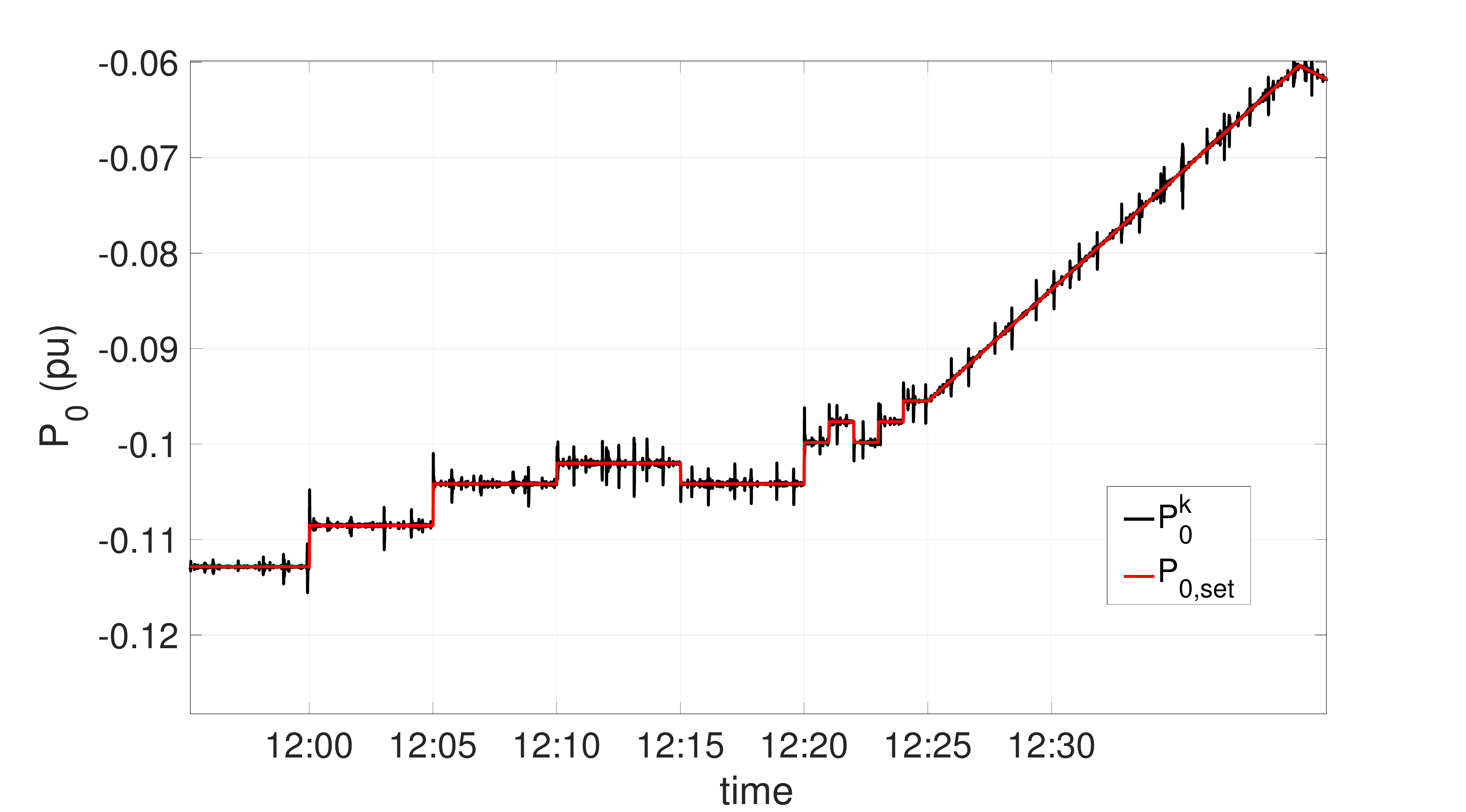}
	\caption{Real power at feeder head during 12:00-12:30.}
	\vspace{-.2cm}
	\label{f2}
\end{figure}
First we specify a given trajectory for the power at the common coupling, which is color-coded in red in Fig. \ref{f2} (negative power indicates reverse power flows). It can be seen that our algorithm is able to regulate $P_0^k$ close to $P_{0,\text{set}}^k$ in real time. Figure \ref{f3} illustrates the voltage profiles for selected nodes. From 10:00 to 12:00 we observe a few flickers, which is caused by rapid variations of the solar irradiance. Other than that, it can be seen that voltage regulation is enforced and a flat voltage profile is obtained. Note that even there are some relatively large jumps from around 12:00 to 14:00, our algorithm is still able to track the optimal trajectory. A comparison with double smoothing algorithm \cite{dall2018optimal} is presented in Figure \ref{f4}. The proposed strategy has potentially better voltage regulation ability, especially for extreme cases e.g. the two spikes from 10:00 to 12:00. 
\begin{figure}
	\centering
	\includegraphics[width=8cm]{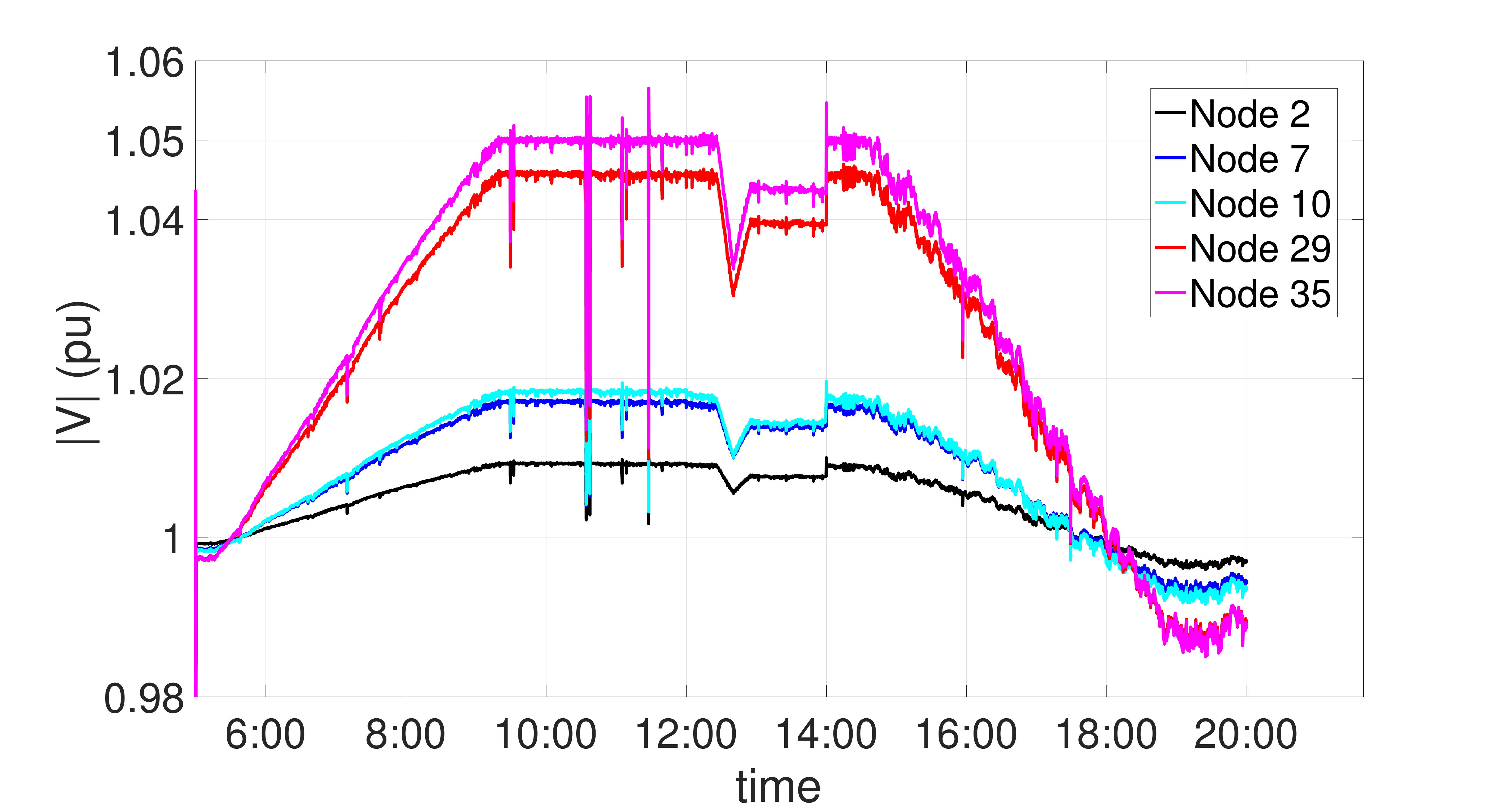}
	\caption{Voltage  profile achieved (only some nodes are considered for illustration purposes). }
	\vspace{-.2cm}
	\label{f3}
\end{figure}
\begin{figure}
	\centering
	\includegraphics[width=8cm]{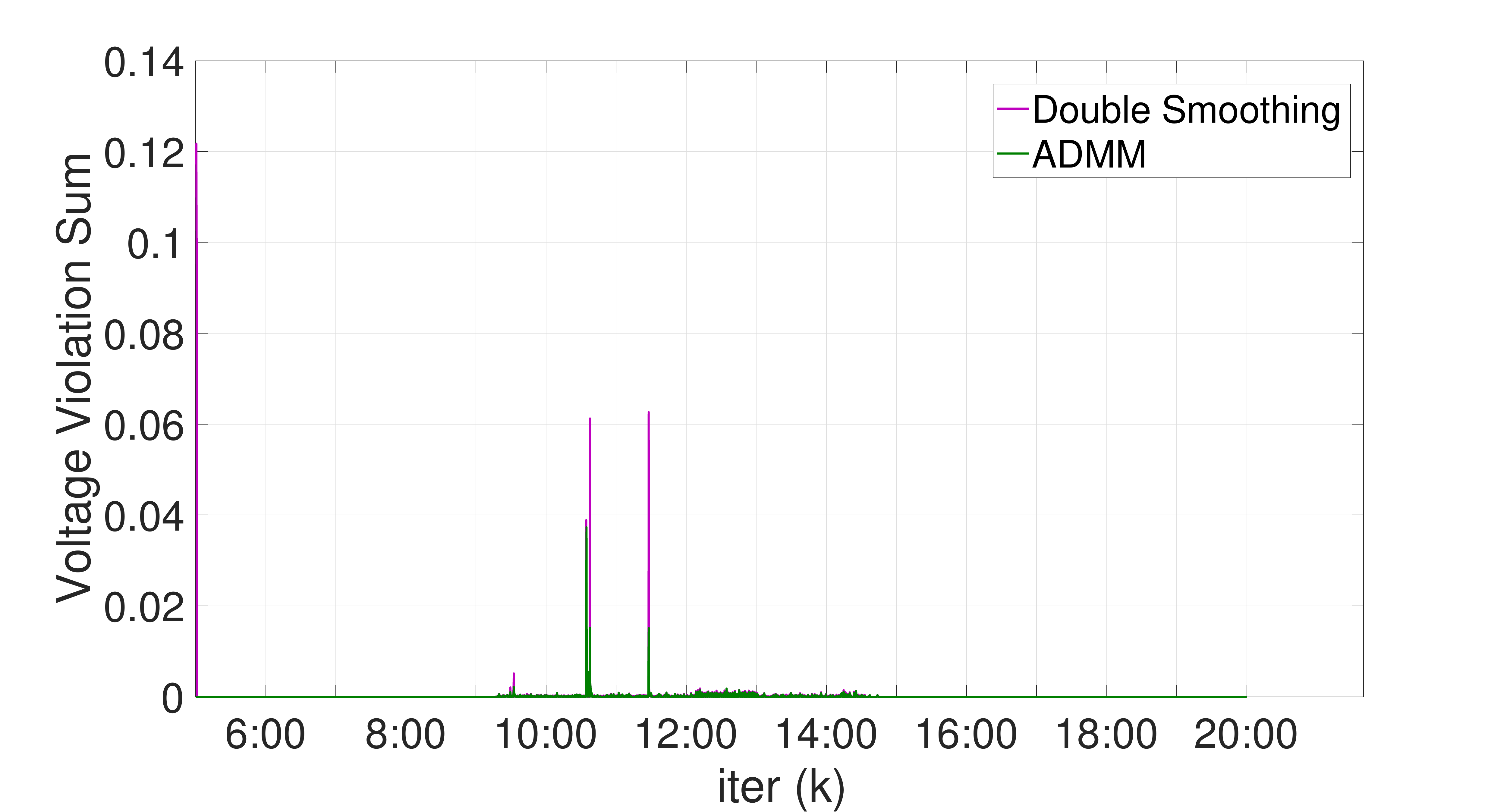}
	\caption{Index for the overall voltage violation across the system $\sum\limits_{n\in\mathcal{N}}\left(\max(|V_n^k| - v^{\text{max}} ,0)+\max(v^{\text{min}} - |V_n^k|,0)\right)$}
	\vspace{-.2cm}
	\label{f4}
\end{figure}
\begin{figure}[t]
        \centering
        \includegraphics[width=8cm]{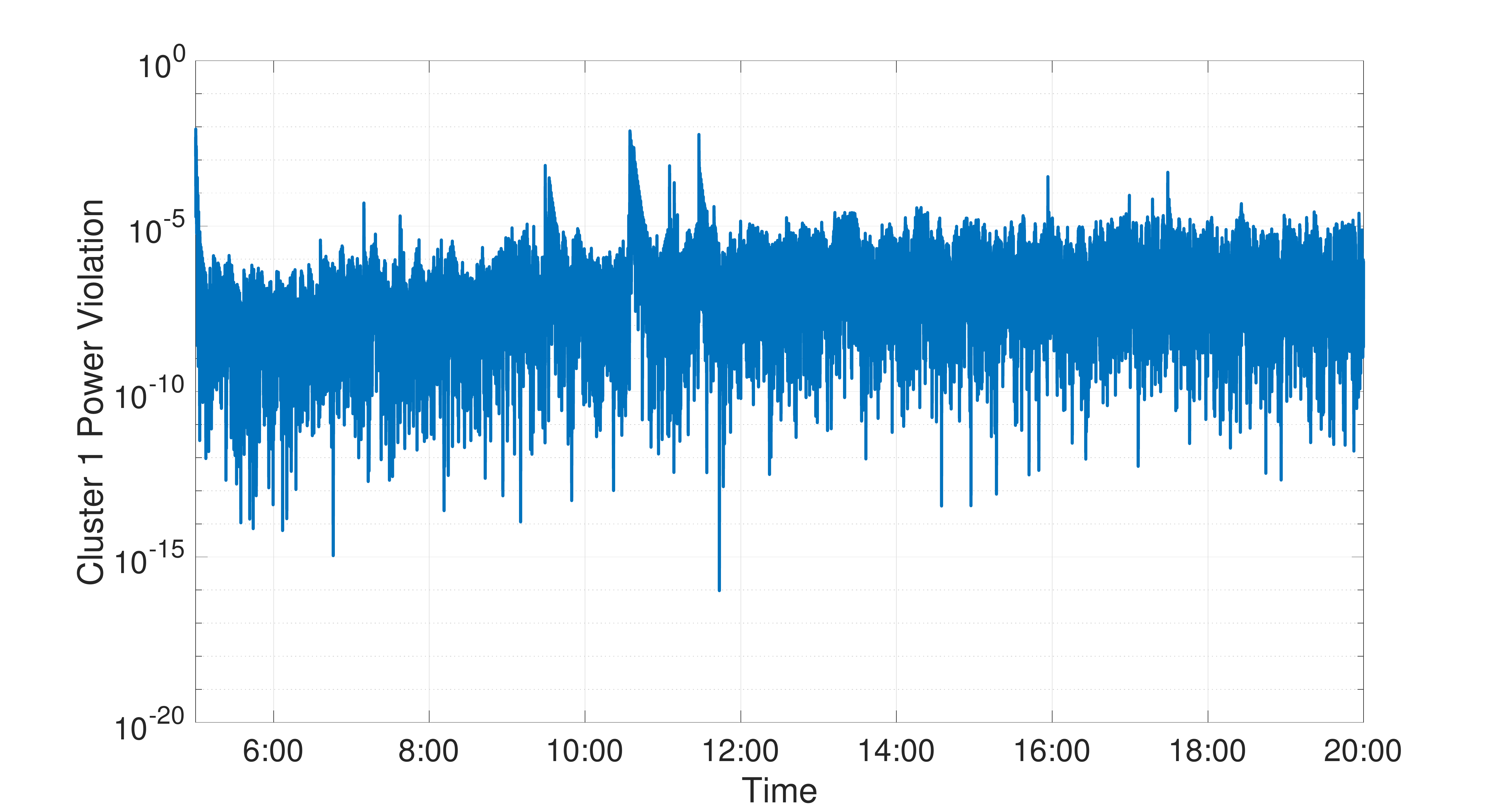}
        \caption{Power violation of Cluster 1: power violation for each cluster is defined as $\|\sum\limits_{i\in\mathcal{N}^{(j)}}\bx^{i\rightarrow j} - \sum\limits_{j\in C}\bx^{j}\|^2$}
        \label{power_vio}
    \end{figure}%
\begin{figure}[t]
        \centering
        \includegraphics[width=8cm]{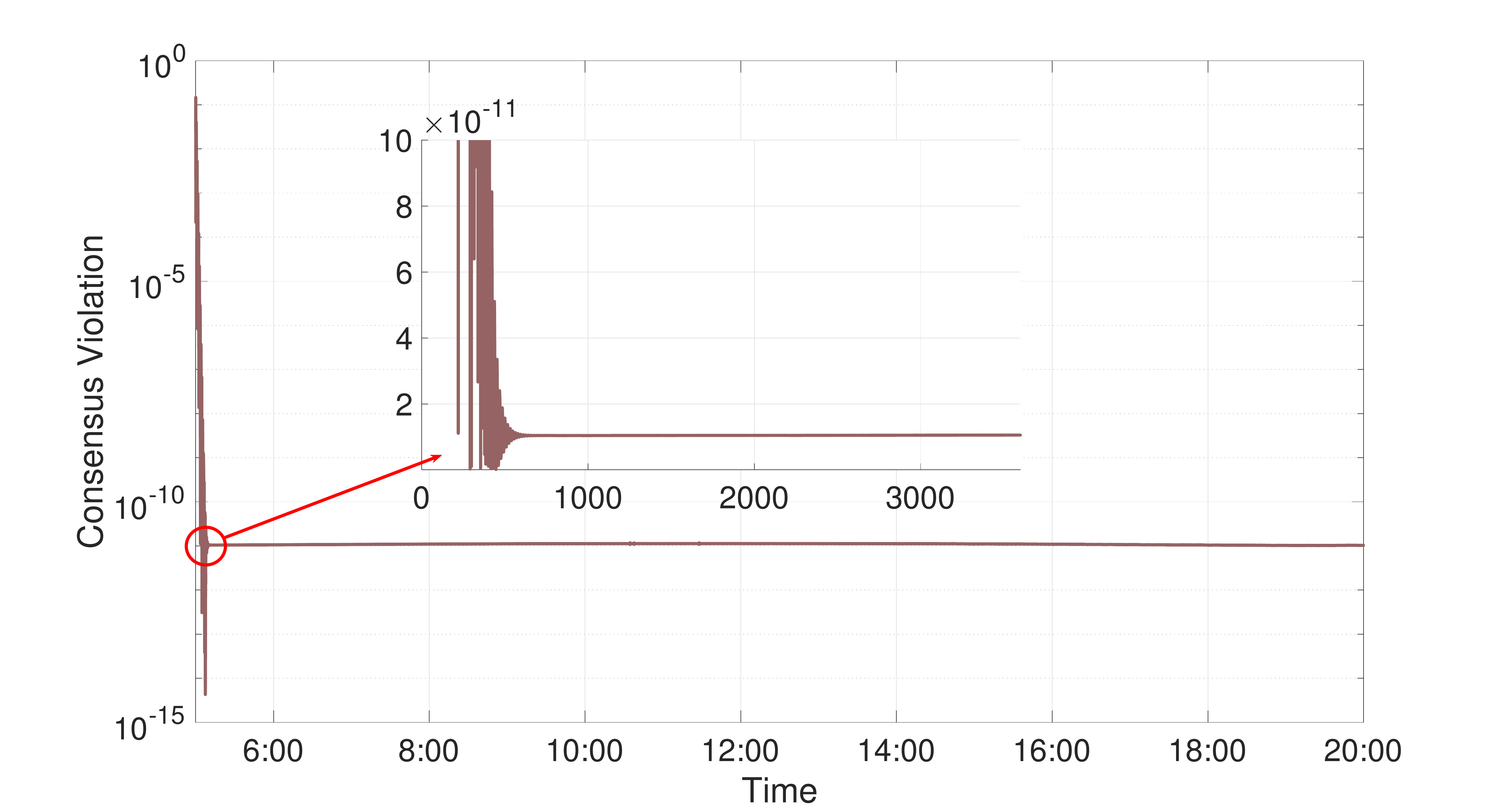}
        \caption{Consensus violation: $\|\bx^{(i\rightarrow j)}+\bx^{(j\rightarrow i)}\|^2$}
        \label{cons_power}
    \end{figure}
We proceed to test in the same setting except we are adding consensus constraints. In Figure \ref{power_vio} we can see that for all 4 clusters, power violation decreases dramatically in first a few minutes and remains at a  low level of $10^{-10}$. The power consensus violation is shown in Figure \ref{cons_power}, where a steep drop at the begging and flat low line after that are observed.

\section{Conclusion}
This paper gives a general online optimization problem formulation and proposes a online algorithm based on alternating direction method of multipliers that can continuously track optimal solution in real time. The steps of ADMM are proximal gradient steps with modification of adding perturbation to dual variable and incorporating system feedback for certain applications. The resulting algorithm is proved to converge to a neighborhood of optimal solution for each time instance. Numerical results for power systems applications also demonstrate the practicality of the proposed algorithm. Our future research will focus on general online nonconvex optimization problems.

\appendices

\subsection{Proof of Theorem \ref{main}}

To show the result of Theorem~\ref{main}, we start from the following lemma. 

\begin{lemma}
\label{lem:linear}
Under the assumptions of Theorem~\ref{main}, it holds that:
\begin{align}
\label{eqn:lin}
\|\bw^{(k)}-\bw^{*,(k)}\|^2_\bG \geq (1+\delta)\|\bw^{(k+1)}-\bw^{*,(k)}\|^2_\bG .
\end{align}

\end{lemma}

\begin{proof}
Since the lemma focuses on a particular time instant $k$, we replace $\bw^{*,(k)}$ with $\bw^{*}$ for notation simplicity. Also, to better fit the equations in the columns, we replace the subscript $^{(k)}$ and $^{(k+1)}$ with $^{k}$ and $^{k+1}$, respectively, in the proof.

From the optimality condition of subproblem \eqref{itr:y}, one has that: 
\begin{align}
&\bA^T\bflambda^k(1-\beta\gamma) - \beta \bA^T(\bA\bx^{k} + \bB\by^{k+1}-\bb)\nonumber\\-&\nabla f_1(\bx^k) + \frac{1}{\alpha_1}\bI(\bx^k\hspace{-1mm}-\bx^{k+1})\in\partial f_0(\bx^{k+1}).\nonumber
\end{align}
Rearrange terms we get
\begin{align}
& \bA^{T}\bflambda^{k+1}\hspace{-1mm}+\hspace{-1mm}\beta\bA^T\bA(\bx^{k+1}\hspace{-1mm} - \hspace{-1mm}\bx^{k})\hspace{-1mm} +\hspace{-1mm} \nabla f_1(\bx^{k+1})\hspace{-1mm} -\hspace{-1mm} \nabla f_1(\bx^k) \nonumber\\
+& \frac{1}{\alpha_1}\bI(\bx^k-\bx^{k+1})\in\partial f_0(\bx^{k+1}) + \nabla f_1(\bx^{k+1}).\label{x_k}
\end{align}
Also from optimality condition of subproblem and \eqref{itr:x}, one has that
\begin{align*}
&\bB^T\bflambda^k(1\hspace{-1mm}-\hspace{-1mm}\beta\gamma)\hspace{-1mm} - \hspace{-1mm} \beta \bB^{T}(\bA\bx^{k+1}\hspace{-1mm}+\hspace{-1mm} \bB\by^{k+1}\hspace{-1mm} -\hspace{-1mm} \bb)\hspace{-1mm}\nonumber\\
+&\beta\bB^T\bA(\bx^{k+1}-\bx^k) + \beta\bB^T\bB(\by^{k+1}-\by^k)\nonumber\\
+& \frac{1}{\alpha_2}\bI(\by^k\hspace{-1mm} -\hspace{-1mm} \by^{k+1}) -\nabla g_1(\by^k)\in\partial g_0(\by^{k+1}) .\nonumber
\end{align*}
Rearrange terms we get
\begin{align}
& \bB^T(\bflambda^{k+1} + \beta \bA(\bx^{k+1}-\bx^k) +\beta\bB(\by^{k+1}-\by^k))\nonumber\\
&+ \nabla g_1(\by^{k+1}) - \nabla g_1(\by^k)+ \frac{1}{\alpha_2}\bI(\by^k-\by^{k+1}) \nonumber\\
\in &\partial g_0(\by^{k+1}) + \nabla g_1(\by^{k+1})\label{y_k}.
\end{align}
Furthermore, from the dual update \eqref{perturb_dual}, one can obtain:  
\begin{align}
\frac{1}{\beta}(\bflambda^k - \bflambda^{k+1}) = \gamma\bflambda^k + (\bA\bx^{k+1} + \bB\by^{k+1} - \bb),\label{dual}
\end{align}
and, together with optimality condition, one obtains: 
\begin{align}\label{dual_opt2}
\frac{1}{\beta}(\bflambda^k\hspace{-1mm} - \hspace{-1mm}\bflambda^{k+1})\hspace{-1mm} =\hspace{-1mm} \gamma(\bflambda^k\hspace{-1mm}-\hspace{-1mm}\bflambda^*) \hspace{-1mm}+\hspace{-1mm} \bA(\bx^{k+1}\hspace{-2mm}-\hspace{-1mm}\bx^{*})\hspace{-1mm} + \hspace{-1mm}\bB(\by^{k+1}\hspace{-2mm}-\hspace{-1mm}\by^{*})
\end{align}
Since the functions $f = f_0 + f_1$ and $g = g_0+g_1$ are strongly convex, we leverage \eqref{eqn:convexity_f} and \eqref{eqn:convexity_g}
and, by plugging the optimality condition \eqref{x_k} and \eqref{y_k}, we have
\begin{align}
&\langle\bA^{T}(\bflambda^{k+1}\hspace{-1mm}-\hspace{-1mm}\blambda^{*})+\beta\bA^T\bA(\bx^{k+1}\hspace{-1mm} -\hspace{-1mm} \bx^{k})\hspace{-1mm} +\hspace{-1mm} \nabla f_1(\bx^{k+1}) \hspace{-1mm}- \hspace{-1mm}\nabla f_1(\bx^k)\nonumber\\
&+ \frac{1}{\alpha_1}\bI(\bx^k-\bx^{k+1}),\bx^{k+1}-\bx^*\rangle\geq v_f\|\bx^{k+1}-\bx^*\|^2\label{x}\\
&\langle\bB^T(\bflambda^{k+1}-\blambda^* + \beta \bA(\bx^{k+1}-\bx^k) +\beta\bB(\by^{k+1}-\by^k)) \nonumber\\
&+ \nabla g_1(\by^{k+1}) - \nabla g_1(\by^k)+\frac{1}{\alpha_2}\bI(\by^k-\by^{k+1}),\by^{k+1}-\by^*\rangle \nonumber\\
&\geq v_g\|\by^{k+1}-\by^*\|^2\label{y}
\end{align}
Next, add \eqref{x} and \eqref{y} together, define $\Phi = v_f\|\bx^{k+1}-\bx^*\|^2+v_g\|\by^{k+1}-\by^*\|^2$ and plug in \eqref{dual_opt2} to obtain:
\begin{align}
&\langle\beta\bA(\bx^{k+1}\hspace{-1mm}-\hspace{-1mm}\bx^k)\hspace{-1mm}+\hspace{-1mm}\blambda^{k+1}\hspace{-1mm}-\hspace{-1mm}\blambda^*,\frac{1}{\beta}(\blambda^k\hspace{-1mm}-\hspace{-1mm}\blambda^{k+1}))\hspace{-1mm}-\hspace{-1mm}\gamma(\blambda^k\hspace{-1mm}-\hspace{-1mm}\blambda^*) \rangle \nonumber\\
+ &\langle\frac{1}{\alpha_1}(\bx^k\hspace{-1mm}-\hspace{-1mm}\bx^{k+1})+\nabla f_1(\bx^{k+1})-\nabla f_1(\bx^k), \bx^{k+1}-\bx^*\rangle \nonumber\\
+& \langle\nabla g_1(\by^{k+1})\hspace{-1mm}-\hspace{-1mm}\nabla g_1(\by^k)+\beta\bB^T\bB(\by^{k+1}\hspace{-1mm}-\hspace{-1mm}\by^k), \by^{k+1}\hspace{-1mm}-\hspace{-1mm}\by^* \rangle \nonumber\\
+&\langle \frac{1}{\alpha_2}(\by^k-\by^{k+1}), \by^{k+1}-\by^*\rangle 
\geq \Phi. \label{cross_all}
\end{align}
We then proceed to split the cross terms and one can notice that there are similar terms for $\bx,\by,\blambda$ in the following form:
\begin{align*}
    &\langle\blambda^{k+1}-\blambda^*,\frac{1}{\beta}(\blambda^k-\blambda^{k+1})\rangle + \langle \frac{1}{\alpha_1}(\bx^k\hspace{-1mm}-\hspace{-1mm}\bx^{k+1}), \bx^{k+1}\hspace{-1mm}-\hspace{-1mm}\bx^*\rangle\\
    &+\frac{1}{\alpha_2}(\by^k-\by^{k+1}), \by^{k+1}-\by^*\rangle 
\end{align*}
We group them together and define the following quantities: 
$$\bG = \left(\begin{matrix}
\frac{1}{\alpha_1}\bI&0&0\\0&\frac{1}{\alpha_2}\bI&0\\0&0&\frac{1}{\beta}\bI
\end{matrix}\right),\bw^k = \left(\begin{matrix}
\bx^k\\\by^k\\ \bflambda^k
\end{matrix}\right) , \,\, \bw^* = \left(\begin{matrix}
\bx^*\\\by^*\\ \bflambda^*
\end{matrix}\right),$$
so that one can rewrite the inequality as follows: 
\begin{align}
&(\bw^{k+1} - \bw^*)^T\bG(\bw^{k} - \bw^{k+1}) + \gamma\langle \bflambda^*-\bflambda^k, \bflambda^{k+1} - \bflambda^*\rangle\nonumber\\
+&\langle \bflambda^{k}\hspace{-1mm} - \hspace{-1mm}\bflambda^{k+1}, \bA(\bx^{k+1}-\bx^{k})\rangle + \beta\gamma\langle \bflambda^*\hspace{-1mm}-\hspace{-1mm}\bflambda^k, \bA(\bx^{k+1}-\bx^{k})\rangle \nonumber\\
+&\langle\beta\bB^T\bB(\by^{k+1} - \by^k),\by^{k+1}-\by^* \rangle\nonumber\\
+&\langle\bx^{k+1}\hspace{-1mm}-\hspace{-1mm}\bx^*,\nabla f_1(x^{k+1})\hspace{-1mm} -\hspace{-1mm} \nabla f_1(\bx^k)\rangle\nonumber\\
+&\langle\by^{k+1}\hspace{-1mm}-\hspace{-1mm}\by^*,\nabla g_1(y^{k+1})\hspace{-1mm} -\hspace{-1mm} \nabla g_1(\by^k)\rangle\hspace{-1mm}\geq \Phi \, . \label{eqn:crossx}
\end{align}
To tackle the cross terms related only to $\blambda$, we consider the following equality
\begin{align}
\|\ba\hspace{-1mm}-\hspace{-1mm}\bc\|^2_\bG\hspace{-1mm} - \hspace{-1mm}\|\bb\hspace{-1mm}-\hspace{-1mm}\bc\|^2_\bG 
\hspace{-1mm}=\hspace{-1mm} 2(\ba\hspace{-1mm}-\hspace{-1mm}\bc)^T\bG(\ba\hspace{-1mm}-\hspace{-1mm}\bb) - \|\ba\hspace{-1mm}-\hspace{-1mm}\bb\|^2_\bG,\label{eqn:fact}
\end{align} 
and use it in \eqref{eqn:crossx} to arrive at the following inequality:
\begin{align}
&(\bw^{k+1} - \bw^*)^T\bG(\bw^{k} - \bw^{k+1}) \geq \frac{\gamma}{2}\|\blambda^{k+1}-\blambda^*\|^2\nonumber\\
-& \frac{\gamma}{2}\|\blambda^{k}\hspace{-1mm}-\hspace{-1mm}\blambda^{k+1}\|^2\hspace{-1mm} + \hspace{-1mm}\frac{\gamma}{2}\|\blambda^{k}\hspace{-1mm}-\hspace{-1mm}\blambda^*\|^2
\hspace{-1mm}+\hspace{-1mm}\langle \bflambda^{k+1}\hspace{-1mm} - \hspace{-1mm}\bflambda^{k}, \bA(\bx^{k+1}\hspace{-1mm}-\hspace{-1mm}\bx^{k})\rangle\nonumber\\
+& \beta\gamma\langle \bflambda^k\hspace{-1mm}-\hspace{-1mm}\bflambda^*\hspace{-1mm}, \bA(\bx^{k+1}\hspace{-1mm}-\hspace{-1mm}\bx^{k})\rangle\hspace{-1mm}+\hspace{-1mm}\langle\bx^{*}\hspace{-1mm}-\hspace{-1mm}\bx^{k+1}\hspace{-1mm},\nabla\hspace{-1mm} f_1(x^{k+1})\hspace{-1mm} -\hspace{-1mm} \nabla\hspace{-1mm} f_1(\bx^k)\rangle\hspace{-1mm} \nonumber\\
+&\langle\by^{*}\hspace{-1mm}-\hspace{-1mm}\by^{k+1},\nabla g_1(y^{k+1})\hspace{-1mm} -\hspace{-1mm} \nabla g_1(\by^k)\rangle\hspace{-1mm}\nonumber\\
+&\langle\beta\bB^T\bB(\by^{k+1} - \by^k),\by^{k+1}-\by^* \rangle+\Phi.\label{eqn:crossy}
\end{align}
Then, we utilize the Cauchy-Schwarz inequality to bound the following term:
\begin{align}
&\langle \bA(\bx^{k+1} -\bx^k),\bflambda^{k+1}-\bflambda^k\rangle \nonumber\\
\geq& -\frac{1}{2\rho_1}\|\bA(\bx^{k+1}\hspace{-1mm}-\hspace{-1mm}\bx^k)\|^2 \hspace{-1mm}-\hspace{-1mm}\frac{\rho_1}{2}\|\bflambda^{k+1} \hspace{-1mm}-\hspace{-1mm}\bflambda^k\|^2, \forall \rho_1 >0\label{cross3}
\end{align}
We implement same process for the rest of cross terms:
\begin{align*}
&\beta\gamma\langle \bA(\bx^{k+1} -\bx^k),\bflambda^k-\bflambda^*\rangle \nonumber\\
\geq& -\frac{\beta\gamma}{2\rho_2}\|\bA(\bx^{k+1}\hspace{-1mm}-\hspace{-1mm}\bx^k)\|^2\hspace{-1mm} -\hspace{-1mm}\frac{\beta\gamma\rho_2}{2}\|\bflambda^{k} \hspace{-1mm}-\hspace{-1mm} \bflambda^*\|^2, \forall \rho_2 >0\label{cross4}\\
&\langle\beta\bB^T\bB(\by^k- \by^{k+1}),\by^{k+1}-\by^* \rangle\nonumber\\
\geq&-\frac{\beta}{\rho_3}\|\bB(\by^{k+1} - \by^{k})\|^2 - \beta\rho_3\|\bB(\by^{k+1}-\by^*)\|^2\nonumber\\
\geq& -\frac{\beta\sigma_{\max}^2(\bB)}{\rho_3}\|\by^{k+1}\hspace{-1mm} -\hspace{-1mm} \by^{k}\|^2-\beta\rho_3\sigma_{\max}^2(\bB)\|\by^{k+1}\hspace{-1mm}-\hspace{-1mm}\by^*\|^2 \, ,
\end{align*}
where $\sigma_{\max}(\bB)$ is the maximal singular value of $\bB$. The remaining terms in~\eqref{eqn:crossy} are related only to gradients of $f_1, g_1$, which can be bounded as follows: 
\begin{align}
&\langle\bx^{*}-\bx^{k+1},\nabla f_1(x^{k+1}) - \nabla f_1(\bx^k)\rangle
\nonumber\\
&+\langle\by^{*}-\by^{k+1},\nabla g_1(y^{k+1}) - \nabla g_1(\by^k)\rangle\nonumber \\
\geq& -\frac{L_f^2}{2\rho_4}\|\bx^{k+1}-\bx^k\|^2 - \frac{\rho_4}{2}\|\bx^{k+1}-\bx^*\|^2 \nonumber\\
&-\frac{L_g^2}{2\rho_5}\|\by^{k+1}-\by^k\|^2 - \frac{\rho_5}{2}\|\by^{k+1}-\by^*\|^2,
\end{align}
where we have used the Cauchy-Schwarz inequality and we leveraged the Lipschitz continuity of $f_1,g_1$. Also, from \eqref{eqn:fact}, it can be noticed that:
\begin{align*}
&\|\bw^k-\bw^*\|^2_\bG - \|\bw^{k+1}-\bw^*\|^2_\bG \\
= ~&2(\bw^k-\bw^*)^TG(\bw^k-\bw^{k+1}) - \|\bw^{k} - \bw^{k+1}\|^2_\bG \, .
\end{align*}
It therefore follows that:
\begin{align*}
&\|\bw^k\hspace{-1mm}-\hspace{-1mm}\bw^*\|^2_\bG \hspace{-1mm}-\hspace{-1mm} \|\bw^{k+1}\hspace{-1mm}-\hspace{-1mm}\bw^*\|^2_\bG  \nonumber\\
\geq& \|\bw^k\hspace{-1mm}-\hspace{-1mm}\bw^{k+1}\|^2_\bG + \gamma\|\blambda^{k+1}-\blambda^*\|^2 \nonumber\\
-& \gamma\|\blambda^{k}-\blambda^{k+1}\|^2 + \gamma\|\blambda^{k}-\blambda^*\|^2\nonumber\\
-&\frac{\sigma_{\max}^2(\bA)}{\rho_1}\|\bx^{k+1}\hspace{-1mm}-\hspace{-1mm}\bx^k\|^2 \hspace{-1mm}-\hspace{-1mm}\rho_1\|\bflambda^{k+1} \hspace{-1mm}-\hspace{-1mm}\bflambda^k\|^2\nonumber\\
-&\frac{\beta\gamma\sigma_{\max}^2(\bA)}{\rho_2}\|\bx^{k+1}\hspace{-1mm}-\hspace{-1mm}\bx^k\|^2\hspace{-1mm} -\hspace{-1mm}\beta\gamma\rho_2\|\bflambda^{k} \hspace{-1mm}-\hspace{-1mm} \bflambda^*\|^2\nonumber\\
-&\frac{L_f^2}{\rho_4}\|\bx^{k+1}-\bx^k\|^2 - \rho_4\|\bx^{k+1}-\bx^*\|^2 \nonumber\\
-&\frac{L_g^2}{\rho_5}\|\by^{k+1}-\by^k\|^2 - \rho_5\|\by^{k+1}-\by^*\|^2+ \Phi\nonumber \\
-&\frac{\beta\sigma_{\max}^2(\bB)}{\rho_3}\|\by^{k+1}\hspace{-1mm} -\hspace{-1mm} \by^{k}\|^2\hspace{-1mm}-\hspace{-1mm}\beta\rho_3\sigma_{\max}^2(\bB)\|\by^{k+1}\hspace{-1mm}-\hspace{-1mm}\by^*\|^2.
\end{align*}
Rearranging the terms in a suitable way, we arrive at the following inequality: 
\begin{align}
&\|\bw^k\hspace{-1mm}-\hspace{-1mm}\bw^*\|^2_\bG \hspace{-1mm}-\hspace{-1mm} \|\bw^{k+1}\hspace{-1mm}-\hspace{-1mm}\bw^*\|^2_\bG 
\nonumber\\
\geq& (\frac{1}{\alpha_1}\hspace{-1mm}-\hspace{-1mm} \frac{\sigma_{\max}^2(\bA)}{\rho_1} \hspace{-1mm}-\hspace{-1mm} \frac{\beta\gamma\sigma_{\max}^2(\bA)}{\rho_2} - \frac{L_f^2}{\rho_4})\|\bx^k-\bx^{k+1}\|^2\nonumber\\
+&(\frac{1}{\alpha_2} - \beta\frac{\sigma_{\max}^2(\bB)}{\rho_3} - \frac{L_g^2}{\rho_5})\|\by^k-\by^{k+1}\|^2\nonumber\\
+&(\frac{1}{\beta}-\gamma-\rho_1)\|\blambda^k-\blambda^{k+1}\|^2+(2v_f-\rho_4)\|\bx^{k+1}-\bx^*\|^2\nonumber\\
+&(2v_g-\rho_5-\beta\rho_3\sigma_{\max}^2(\bB))\|\by^{k+1}-\by^*\|^2 \nonumber\\
+&\gamma\|\blambda^{k+1}-\blambda^*\|^2 + (\gamma-\beta\gamma\rho_2)\|\blambda^k-\blambda^*\|^2\label{cross2}.
\end{align}
Recall that the goal is to prove the following inequality:
\begin{align}
\|\bw^k-\bw^*\|^2_\bG \geq (1+\delta)\|\bw^{k+1}-\bw^*\|^2_\bG,\label{linconv1}
\end{align}
where $\delta>0$ is a constant. For brevity, denote the right-hand-side of \eqref{cross2} as $C$; then it is sufficient to prove that:
\begin{align*}
&C\geq \delta\|\bw^{k+1}\hspace{-1mm}-\hspace{-1mm}\bw^*\|^2_\bG \nonumber\\
&\hspace{-3mm}=\hspace{-1mm}\frac{\delta}{\alpha_1}\|\bx^{k+1}\hspace{-1mm}-\hspace{-1mm}\bx^*\|^2 \hspace{-1mm}+\hspace{-1mm} \frac{\delta}{\alpha_2}\|\by^{k+1}\hspace{-1mm}-\hspace{-1mm}\by^*\|^2\hspace{-1mm}+\hspace{-1mm}\frac{\delta}{\beta}\|\bflambda^{k+1}\hspace{-1mm}-\hspace{-1mm}\bflambda^*\|^2 \, 
\end{align*}
which  requires  the following to hold true:
\begin{align*}
&\frac{1}{\alpha_1} - \frac{\max\sigma^2(\bA)}{\rho_1} - \frac{\beta\gamma\max\sigma^2(\bA)}{\rho_2} - \frac{L_f^2}{\rho_4} \geq 0\\
&\frac{1}{\alpha_2} - \frac{\beta\max\sigma^2(\bB)}{\rho_3} - \frac{L_g^2}{\rho_5}\geq 0,~ \frac{1}{\beta}-\gamma-\rho_1 \geq 0 \\
&2v_f-\rho_4-\frac{\delta}{\alpha_1} \geq 0,~ 2v_g - \rho_5-\beta\rho_3\max\sigma^2(\bB) -\frac{\delta}{\alpha_2} \geq 0\\
&\gamma-\frac{\delta}{\beta} \geq 0,~ \gamma-\beta\gamma\rho_2\geq 0.
\end{align*}
From the inequalities above, one can notice  that the constant $\alpha_1,\alpha_2$ is closely related to various constants as well as the  singular values of $\bA$ and $\bB$, denoted as $\sigma(\bA)$ and $\sigma(\bB)$, respectively. Specifically, this leads to the following conditions for the step sizes:
\begin{align}
&\frac{\delta}{2v_f\hspace{-1mm}-\hspace{-1mm}\rho_4} \leq \alpha_1 \leq \frac{1}{(\frac{1}{\rho_1}\hspace{-1mm}+\hspace{-1mm}\frac{\beta\gamma}{\rho_2})\max\sigma^2(\bA)\hspace{-1mm} + \hspace{-1mm}\frac{L_f^2}{\rho_4}},  \label{eqn:alpha1}\\
&\frac{\delta}{2v_g\hspace{-1mm}-\hspace{-1mm}\rho_5\hspace{-1mm}-\hspace{-1mm}\beta\rho_3\max\sigma^2(\bB)}\leq \alpha_2\leq \frac{1}{\frac{\beta\max\sigma^2(\bB)}{\rho_3} \hspace{-1mm}+\hspace{-1mm} \frac{L_g^2}{\rho_5}} \, . \label{eqn:alpha2}
\end{align}
To ensure that there exists step sizes $\alpha_1,\alpha_2$ that satisfy the condition above, one can choose $\rho_1 = 1, \rho_2 = 1, \rho_3 = \frac{v_g}{2\beta\max\sigma^2(\bB)}, \rho_4 = v_f, \rho_5 = v_g$, and from Assumption \ref{Assum_Bounded},\ref{Assum_uniform} we know we have all problem dependent parameters uniform bounded. 
As for other parameters, one has that: 
\begin{align*}
\frac{1}{\beta}-\gamma-1 \geq 0,~ \gamma-\beta\gamma\geq 0\Rightarrow~ \beta\gamma+\beta \leq 1,~ \beta\leq 1.
\end{align*}
This completes the proof.
\end{proof} 
Using  Lemma~\ref{lem:linear}, we can now proven the result of  Theorem~\ref{main}.  For notation simplicity, we define $r = \frac{1}{1+\delta}$. Using~\eqref{eqn:lin} and the triangle inequality, we have that
	\begin{align*}
	&\|\bw^{(k)} \hspace{-1mm}-\hspace{-1mm} \bw^{*,(k)}\|_\bG \leq r\|\bw^{(k-1)}\hspace{-1mm} - \hspace{-1mm} \bw^{*,(k-1)} \hspace{-1mm}+\hspace{-1mm} \bw^{*,(k-1)} \hspace{-1mm}-\hspace{-1mm} \bw^{*,(k)}\|_\bG \\
	& \leq r\|\bw^{(k-1)}\hspace{-1mm} - \hspace{-1mm} \bw^{*,(k-1)}\|_\bG +  r\|\bw^{*,(k-1)} - \bw^{*,(k)}\|_\bG \, .
	\end{align*} 
We now find a bound on $\|\bw^{*,(k-1)} \hspace{-1mm}-\hspace{-1mm} \bw^{*,(k)}\|_\bG$. 
From \eqref{diff}, \eqref{dual_opt} and the fact that we have chosen $\gamma=1$, we know that 
	   \begin{align*}
		&\bA^{(k+1)}\bx^{*,(k+1)}\hspace{-1mm}-\hspace{-1mm}\bA^{(k)}\bx^{*,(k)}\hspace{-1mm} +\hspace{-1mm} \bB^{(k+1)}\by^{*,(k+1)}\hspace{-1mm}-\hspace{-1mm}\bB^{(k)}\by^{*,(k)}\nonumber\\
		+&\bb^{(k)}-\bb^{(k+1)}+ \blambda^{*,(k+1)} -\blambda^{*,(k)} = 0.
		\end{align*}
Move $\blambda$ terms to the other side of the equation and take norm for both hand sides, it follows:
		\begin{align*}
		&\|\blambda^{*,(k+1)} - \blambda^{*,(k)}\| \\
        \leq& \|\bA^{(k+1)}(\bx^{*,(k+1)} \hspace{-1mm}-\hspace{-1mm} \bx^{*,(k)}) + (\bA^{(k+1)}-\bA^{(k)})\bx^{*,(k)}\| \\
        &+\|\bB^{(k+1)}(\by^{*,(k+1)} \hspace{-1mm}-\hspace{-1mm} \by^{*,(k)}) + (\bB^{(k+1)}-\bB^{(k)})\by^{*,(k)}\|\\
        &+\|\bb^{(k)}-\bb^{(k+1)}\|.
        \end{align*}
        From triangle inequality of norm we know
        \begin{align*}
        &\|\blambda^{*,(k+1)} - \blambda^{*,(k)}\| \\
        \leq& \|\bA^{(k+1)}\|\|\bx^{*,(k+1)} \hspace{-1mm}-\hspace{-1mm} \bx^{*,(k)}\|\hspace{-1mm} +\hspace{-1mm} \|\bB^{(k+1)}\|\|\by^{*,(k+1)}\hspace{-1mm} -\hspace{-1mm} \by^{*,(k)}\| \\ &+ \|\bb^{(k)}-\bb^{(k+1)}\|\hspace{-1mm}+\hspace{-1mm} \|\bA^{(k+1)}-\bA^{(k)}\|\|\bx^{*,(k)}\|\\
		&+\|\bB^{(k+1)}-\bB^{(k)}\|\|\by^{*,(k)}\|
		\end{align*}
		From \eqref{nedic2} we know that 
		\begin{align*}
		   \|\bx^{*,(k)}\| &\leq \sigma_1 + \sqrt{\gamma}\max\|\blambda^{\text{opt}, (k)}\|\cdot c \\
		    &\leq \sigma_1 + \sqrt{\gamma}\mathcal{M}c:= \mathcal{J}(\sigma_1), \\
		    \|\by^{*,(k)}\|&\leq \sigma_2 + \sqrt{\gamma}\max\|\blambda^{\text{opt}, (k)}\|\cdot c\\
		    &\leq \sigma_2 + \sqrt{\gamma}\mathcal{M}c:= \mathcal{J}(\sigma_2).
		\end{align*}
		Combining with Assumption~\ref{Assum_Bounded} we can reach the following inequality
		\begin{align*}
		&\|\blambda^{*,(k+1)} - \blambda^{*,(k)}\| \\
		\leq&\tilde{\sigma}_{\bA}\sigma_\bx+\tilde{\sigma}_{\bB}\sigma_\by+\sigma_\bb+\sigma_\bA\mathcal{x}(\sigma_1) + \sigma_\bB\mathcal{x}(\sigma_2)
		\triangleq \sigma_{\blambda},
	\end{align*}
which gives us 
\begin{align}
	\|\bw^{*,(k-1)} - \bw^{*,(k)}\|_\bG \leq \psi :=  \sqrt[]{\frac{\sigma_\bx^2}{\alpha_1} +\frac{\sigma_\by^2}{\alpha_2} + 2\sigma_{\blambda}^2}
	\end{align}
and the desired result is obtained.

\subsection{Proof of Corollary \ref{linear}}

By recursively applying the  Theorem~\ref{main}, we have that
\begin{align*}
	\|\bw^{(k)} \hspace{-1mm}-\hspace{-1mm} \bw^{*,(k)}\|_\bG 
    & \leq r^k\|\bw^{(0)}-\bw^{*,(0)}\|_\bG \\
    & +  \sum_{i=1}^{(k)}r^{k-i+1}\|\bw^{*,(i-1)}-\bw^{*,(i)}\|_\bG
\end{align*}
where $\|\bw^{*,(k-1)} - \bw^{*,(k)}\|_\bG \leq \psi$.  Taking $k\rightarrow+\infty$, we can derive
	\begin{align*}
	&\lim_{k\rightarrow\infty}\|\bw^{(k)} - \bw^{*,(k)}\|_\bG  \nonumber\\
    \leq&\lim_{k\rightarrow\infty}\left( \frac{r(1-r^{(k)})}{1-r}\psi(\sigma) + r^k\|\bw^{(0)}-\bw^{*,(0)}\|_\bG\right)\\
	\Rightarrow &\lim\sup_{k\rightarrow\infty}\|\bw^{(k)} - \bw^{*,(k)}\|_\bG \leq \frac{r}{1-r}\psi = \frac{\psi}{\delta}.
	\end{align*}
	The desired result is then obtained.

\subsection{Remarks on Corollary 2}
The result is obtained by choosing the biggest step sizes;  that is:
\begin{align*}
    \alpha_1 = \frac{1}{(1+\beta\gamma)\tilde{\sigma^2}_\bA + \frac{\tilde{L}_f^2}{\tilde{v}_f}}, ~\alpha_2 = \frac{1}{\frac{2\beta^2\max\tilde{\sigma}^4_\bB}{\tilde{v}_g}+ \frac{\tilde{L}_g^2}{\tilde{v}_g}}
\end{align*}
Recall that:
\begin{align}
	    \delta &\leq \frac{\tilde{v}_f}{(1+\beta\gamma)\tilde{\sigma}^2_\bA + \frac{\tilde{L}_f^2}{\tilde{v}_f}},~\delta \leq \frac{\tilde{v}_g}{\frac{4\beta^2\tilde{\sigma}^4_\bB}{\tilde{v}_g} + \frac{2\tilde{L}_g^2}{\tilde{v}_g}}.\label{eqn:delta}
\end{align}
We already know that $\delta \leq \beta\gamma$; therefore, based on~\eqref{eqn:alpha1}--\eqref{eqn:alpha2}, one can pick $\delta$ as 
\begin{align*}
\delta = \min \left(\frac{\tilde{v}_f}{(1+\beta\gamma)\tilde{\sigma}^2_\bA + \frac{\tilde{L}_f^2}{\tilde{v}_f}},\frac{\tilde{v}_g}{\frac{4\beta^2\tilde{\sigma}^4_\bB}{\tilde{v}_g} + \frac{2\tilde{L}_g^2}{\tilde{v}_g}}, \beta\gamma \right).
\end{align*}.

\ifCLASSOPTIONcaptionsoff
  \newpage
\fi

\bibliographystyle{IEEEtran}
\bibliography{2020-arxiv}{}

\end{document}